\documentclass{article}
\usepackage{graphicx} 
\usepackage{amsmath}
\usepackage{enumitem}
\PassOptionsToPackage{unicode}{hyperref}
\PassOptionsToPackage{hyphens}{url}

\usepackage{amsmath,amssymb, amsthm}
\usepackage[a4paper, margin=1in]{geometry}
\usepackage{lmodern}
\usepackage{subcaption}
\usepackage{comment}
\usepackage{iftex}
\usepackage{authblk}
\ifPDFTeX
  \usepackage[T1]{fontenc}
  \usepackage[utf8]{inputenc}
  \usepackage{textcomp} 
\else 
  \usepackage{unicode-math}
  \defaultfontfeatures{Scale=MatchLowercase}
  \defaultfontfeatures[\rmfamily]{Ligatures=TeX,Scale=1}
\fi
\IfFileExists{upquote.sty}{\usepackage{upquote}}{}
\IfFileExists{microtype.sty}{
  \usepackage[]{microtype}
  \UseMicrotypeSet[protrusion]{basicmath} 
}{}
\makeatletter
\@ifundefined{KOMAClassName}{
  \IfFileExists{parskip.sty}{%
    \usepackage{parskip}
  }{
    \setlength{\parindent}{0pt}
    \setlength{\parskip}{6pt plus 2pt minus 1pt}}
}{
  \KOMAoptions{parskip=half}}
\makeatother
\usepackage{xcolor}
\usepackage{longtable,booktabs,array}
\usepackage{calc} 
\usepackage{etoolbox}
\makeatletter
\patchcmd\longtable{\par}{\if@noskipsec\mbox{}\fi\par}{}{}
\makeatother
\IfFileExists{footnotehyper.sty}{\usepackage{footnotehyper}}{\usepackage{footnote}}
\makesavenoteenv{longtable}
\usepackage{graphicx}
\makeatletter
\def\maxwidth{\ifdim\Gin@nat@width>\linewidth\linewidth\else\Gin@nat@width\fi}
\def\maxheight{\ifdim\Gin@nat@height>\textheight\textheight\else\Gin@nat@height\fi}
\makeatother
\setkeys{Gin}{width=\maxwidth,height=\maxheight,keepaspectratio}
\makeatletter
\def\fps@figure{htbp}
\makeatother
\usepackage[normalem]{ulem}
\ifLuaTeX
  \usepackage{selnolig}  
\fi
\IfFileExists{bookmark.sty}{\usepackage{bookmark}}{\usepackage{hyperref}}
\IfFileExists{xurl.sty}{\usepackage{xurl}}{} 
\hypersetup{
  hidelinks,
  pdfcreator={LaTeX via pandoc}}

\newtheorem{theorem}{Theorem}

\author[]{Jungin E. Kim}
\author[]{Yan Wang}
\affil[]{Georgia Institute of Technology, Atlanta, GA 30332, USA 
\\E-mail: {yan-wang@gatech.edu} }

\title{Variational Quantum Algorithm for Constrained Topology Optimization}

\begin{document}

\maketitle

\section*{Abstract}
One of the challenging scientific computing problems is topology optimization, where searching through the combinatorially complex configurations and solving the constraints of partial differential equations need to be done simultaneously. In this paper, a novel variational quantum algorithm for constrained topology optimization is proposed, which allows for the single-loop parallel search for the optimal configuration that also satisfies the physical constraints. The optimal configurations and the solutions to physical constraints are encoded with two separate registers.  A constraint encoding scheme is also proposed to incorporate volume and connectivity constraints in optimization. The gate complexity of the proposed quantum algorithm is analyzed. The algorithm is demonstrated with compliance minimization problems including truss structures and Messerschmitt-B\"{o}lkow-Blohm beams.

\section{Introduction}
\label{sec:introduction}

Topology optimization (TO) is the process of finding the optimal material configuration of a structure to achieve some desirable properties \cite{bendsoe2013topology}. It has a wide variety of applications, such as microfluidic reactors \cite{okkels2007scaling}, photonic crystals \cite{lin2016enhanced}, topological insulators \cite{christiansen2019topological}, cloaking \cite{fujii2020dc},  acoustic barriers \cite{xu2020topology}, light-weight structures \cite{hagishita2009topology}, multi-material composites \cite{blasques2012multi}, and many more. 
However, solving TO problems is challenging for two reasons. First, searching through a combinatorially large number of possible configurations is expensive. TO researchers developed density-based methods to convert the discrete problem to the optimization with continuous variables so that classical optimizers can be applied. Yet, the challenge of high dimensionality for the searching space remains.
Second, the evaluation of design objective for each configuration requires solving partial differential equations (PDEs) to obtain the quantities of interest and often relies on expensive simulations.  Therefore, computational efficiency remains a major challenge in TO.

In the past decade, quantum scientific computing has emerged as a new paradigm to solve difficult scientific computing problems on quantum computers, such as engineering simulation \cite{wang2013simulating, wang2016accelerating}, optimization \cite{ wang2014global, wang2023opportunities}, materials design \cite{bauer2020quantum,von2021quantum}, drug discovery \cite{cao2018potential}, and others.  It has the potential to significantly improve the efficiency of TO, because each quantum basis state can naturally encode a configuration in the exponentially large design solution space \cite{wang2023opportunities}. 

Recently, both quantum annealers and circuit-based quantum computers started being used to solve TO problems. In quantum annealing \cite{morita2008mathematical},  problems are formulated as quadratic unconstrained binary optimization (QUBO). The quantum system is perturbed in order to transition from the initial state to the target ground state,  utilizing the adiabatic quantum computation scheme \cite{albash2018adiabatic}. The ground state represents the optimal solution. Ye et al. \cite{ye2023quantum} decomposed TO problems into sequences of mixed-integer nonlinear programs, where each problem is reformulated as QUBO. Honda et al. \cite{honda2024development} optimized the topology of two- and three-dimensional truss structures, where the nodal displacements are expressed as binary strings. Although QUBO formulations allow quantum annealing to be applied to encode material configurations efficiently, the PDE constraints in TO involve continuous variables and impose the challenge of binary representation to quantum annealers. Among the circuit-based computational approaches, variational quantum algorithms (VQAs) such as quantum approximate optimization algorithm (QAOA) \cite{farhi2014quantum, hadfield2019quantum} and variational quantum eigensolver \cite{peruzzo2014variational} are the most suitable ones to solve real-world problems on near-term quantum computers. 
In these hybrid quantum-classical approaches, computation is performed by evolving the quantum state through parameterized quantum gates. The parameters are optimized by classical optimizers iteratively in the outer loop. 
Kim and Wang \cite{kim2023topology} applied a quantum approximate Bayesian optimization algorithm \cite{kim2023quantum} to perform TO for trusses and metamaterials.  A surrogate-based global optimizer, Bayesian optimization, is used to train the parameters of the variational quantum circuit. Sato et al. \cite{sato2023quantum} proposed a double-loop optimization framework which involves two quantum circuits. The first circuit is used to solve the PDE constraints for all possible configurations. The second circuit is to maximize the probability of obtaining the optimal configuration. 

In this paper, we propose a new VQA to solve TO problems efficiently through quantum entanglement. Two quantum registers are used to encode the two types of values in TO, respectively.  The first type is the optimal configuration as the solution of TO, whereas the second one represents physical quantities which satisfy the PDE constraints. Two entanglement mechanisms are introduced to compute both types of values concurrently. First, the problem Hamiltonian operators corresponding to the PDE constraints in the VQA circuit are entangled with the first register. Second, the qubits in the second register are also entangled to improve the searching efficiency and reduce the circuit depth. In addition, a constraint encoding scheme is proposed to incorporate connectivity and volume constraints of the optimization problems as constraints of VQA parameters. 

Researchers have developed approaches to include constraints in quantum optimization. The most common approach is to incorporate constraints as the regularization or penalty terms in the objective function. The problem Hamiltonian is subsequently altered, and the ground state corresponds to the optimal solution which is likely to satisfy the constraints. This approach is for problems involving soft constraints. For instance, Hen and Sarandy \cite{hen2016driver} formulated the Hamiltonian as the sum of the original Hamiltonian and the operators that encode the constraints. Fern\'andez-Lorenzo et al. \cite{fernandez2021hybrid} used the variational quantum eigensolver and QAOA to find the ground state of an Ising Hamitonian which represents the Lagrangian of a QUBO problem.  Djidjev \cite{djidjev2023logical} used quantum annealing to minimize a sum of objective, constraint, and the squared constraint functions. Monta\~nez-Barrera et al. \cite{montanez2024unbalanced} solved QUBO problems with quantum annealing and QAOA, where the constraint penalty term is formulated as an exponential function. A second approach to incorporate constraints is to design quantum operators that only search within the feasible solution space. In QAOA, this is achieved by designing mixing Hamiltonians. One example is the XY mixing Hamiltonian which restricts the feasible basis states on a graph \cite{wang2020xy, hadfield2022analytical}. The second example is to design the mixing Hamiltonian in QAOA or the driver Hamiltonian for quantum annealing such that it commutes with the operators associated with linear constraints \cite{leipold2021constructing}. A third instance is a continuous-time quantum walk operator \cite{wang2021quantum, Qu_2024}, where a graph adjacency matrix restricts the search process. The third approach to incorporate constraints is to treat objective and constraints separately and formulate the problems as multi-objective optimization. For instance, D\'iez-Valle et al. \cite{diezvalle2023multiobjectivevqa} proposed VQAs to solve the objective minimization and constraint satisfaction problems simultaneously.  

The proposed constraint encoding scheme here is different from the existing approaches. The design constraints are mapped to the constraints of the parameters in the variational quantum circuit. Particularly, the volume constraints in TO are enforced by setting an upper bound for the parameters of the uncontrolled rotation gates, whereas the connectivity constraints are encoded in the controlled rotation gates such that the amplitudes of infeasible configurations are minimized.

The remainder of the paper is structured as follows. In Section \ref{sec:methodology}, the proposed TO formalism is introduced. In Section \ref{sec:constraintencoding}, the quantum circuit which encodes the connectivity and volume constraints is described. In Section \ref{sec:results}, the proposed TO approach is demonstrated with five examples.  The first three examples involve minimizing the compliance of truss structures, whereas the other two examples involve minimizing the compliance of a Messerschmitt-B\"{o}lkow-Blohm (MBB) beam. Discussions and conclusions are provided in Section \ref{sec:conclusions}.
\vspace{-8pt}

\section{Proposed Variational Quantum Topology Optimization Algorithm} \label{sec:methodology}

In TO problems, the spatial domain is typically discretized as elements. Choosing the material configuration is equivalent to choosing a subset of elements so that the best objective value is obtained. The proposed VQA performs the tasks of finding the optimal material configuration and solving the PDE constraints simultaneously. The parameterized quantum circuit includes two registers. The first register encodes the material configurations, whereas the second one encodes the PDE solutions. Each PDE solution is entangled with the corresponding material configuration. This is done through a controlled problem Hamiltonian operator, which is implemented with a sequence of controlled operators applied on the second register. A sequence of controlled and uncontrolled rotation gates are also applied so that the qubits in the second register become entangled. The entanglement between the qubits reduces the size of the solution space.

\subsection{Problem Hamiltonian formulation} \label{sec:problemhamiltonian}

Without the loss of generality, the truss compliance minimization problem is used to illustrate the formulation of problem Hamiltonians.
The compliance minimization problem for truss structure design is typically formulated as
\begin{equation} \label{eq:compliancemin}
\begin{gathered}
    \min_{\boldsymbol{q}, \boldsymbol{u}} \boldsymbol{u}^{\intercal} K(\boldsymbol{q})\boldsymbol{u} \\
    \text{s.t.} \quad 
    K(\boldsymbol{q})\boldsymbol{u} = \Tilde{\boldsymbol{f}},
\end{gathered}
\end{equation}
where $\boldsymbol{u} \in \mathbb{R}^N$ is the displacement vector, $N$ is the number of degrees of freedom, $\boldsymbol{q}$ is the material configuration, $K(\boldsymbol{q}) \in \mathbb{R}^{N \times N}$ is the total stiffness matrix of the structure which depends on configuration $\boldsymbol{q}$, and $\Tilde{\boldsymbol{f}} \in \mathbb{R}^N$ is the load vector. The linear constraint in Eq. (\ref{eq:compliancemin}) is obtained from the weak form of PDEs in the finite-element formulation. In the VQA, the discrete variable of material configurations $\boldsymbol{q}$ is encoded as the state of the first quantum register, whereas the displacement vector $\boldsymbol{u}$ is encoded as the amplitudes of the second register.  

In a double-loop TO process, the constraint in Eq. (\ref{eq:compliancemin}) is solved in the inner loop, whereas the optimization is performed in the outer loop. To solve the problem in a single-loop TO process, the constraint in Eq. (\ref{eq:compliancemin}) needs to be combined into the objective as the regularization term. It is important to obtain a quadratic form of objective function so that the Hamiltonian can be constructed. Therefore, the compliance minimization problem is reformulated as
\begin{equation} \label{eq:complianceminlagrangian}
\begin{gathered}
    \min_{\boldsymbol{q}, \boldsymbol{u}, \boldsymbol{f}, \boldsymbol{\alpha}} \boldsymbol{u}^{\intercal} K(\boldsymbol{q})\boldsymbol{u} + \lambda \Vert K(\boldsymbol{q})\boldsymbol{u} - \Gamma(\boldsymbol{\alpha}) \boldsymbol{f} \Vert^2 \\
    \text{s.t.} \quad \Gamma(\boldsymbol{\alpha}) \boldsymbol{f} = \Tilde{\boldsymbol{f}},
\end{gathered}
\end{equation}
where $\boldsymbol{f} \in \mathbb{R}^N$ is a vector of variables corresponding to the estimated loads, $\lambda \in \mathbb{R}^{+}$ is a Lagrange multiplier,  $\boldsymbol{\alpha} \in \mathbb{R}^{N}$ is a vector of scaling factors, and $\Gamma(\boldsymbol{\alpha})=\mathrm{diag}(\alpha_1,\ldots,\alpha_N)$ is a diagonal matrix with the elements of $\boldsymbol{\alpha}$ as its diagonal entries.  $\Gamma$ must be included so that $\boldsymbol{f}$ is scaled to $\Tilde{\boldsymbol{f}}$ (i.e. $\Vert \Gamma\boldsymbol{f} \Vert = \Vert \Tilde{\boldsymbol{f}} \Vert$).

With $\boldsymbol{f}$ introduced into the problem in Eq. (\ref{eq:complianceminlagrangian}), the objective function is now strictly quadratic with respect to $\boldsymbol{w}=\begin{bmatrix}
        \boldsymbol{u}^{\intercal} & -\boldsymbol{f}^{\intercal} \\
    \end{bmatrix}^{\intercal}$, because
\begin{equation} \label{eq:strainenergy_quad}
    \boldsymbol{u}^{\intercal}K(\boldsymbol{q})\boldsymbol{u} = 
    \begin{bmatrix}
        \boldsymbol{u}^{\intercal} & -\boldsymbol{f}^{\intercal} \\
    \end{bmatrix}
    \begin{bmatrix}
        K(\boldsymbol{q}) & \boldsymbol{0} \\
        \boldsymbol{0} & \boldsymbol{0} \\
    \end{bmatrix}
    \begin{bmatrix}
        \boldsymbol{u} \\
        -\boldsymbol{f} \\
    \end{bmatrix}
\end{equation}
and
\begin{equation} 
\label{eq:PDEresidual_quadratic}
\begin{split}
     \Vert K(\boldsymbol{q})\boldsymbol{u} - \Gamma(\boldsymbol{\alpha})\boldsymbol{f} \Vert^2
     =
    \begin{bmatrix}
        \boldsymbol{u}^{\intercal} & -\boldsymbol{f}^{\intercal} \\
    \end{bmatrix}
    \begin{bmatrix}
         K^{2}(\boldsymbol{q}) &  K(\boldsymbol{q}) \Gamma(\boldsymbol{\alpha})  \\
         \Gamma(\boldsymbol{\alpha})K(\boldsymbol{q}) &  \Gamma^2(\boldsymbol{\alpha}) \\
    \end{bmatrix}
    \begin{bmatrix}
        \boldsymbol{u} \\
        -\boldsymbol{f} \\
    \end{bmatrix}.
\end{split}
\end{equation}
Let
\begin{equation} \label{eq:L_matrix}
    A(\boldsymbol{q}) = 
    \begin{bmatrix}
        K(\boldsymbol{q}) & \boldsymbol{0} \\
        \boldsymbol{0} & \boldsymbol{0} \\
    \end{bmatrix}
    \in \mathbb{R}^{2N \times 2N}
\end{equation}
and
\begin{equation} \label{eq:M_matrix}
    B(\boldsymbol{q}, \boldsymbol{\alpha}) = 
    \begin{bmatrix}
         K^{2}(\boldsymbol{q}) & K(\boldsymbol{q}) \Gamma(\boldsymbol{\alpha})  \\
         \Gamma(\boldsymbol{\alpha}) K(\boldsymbol{q}) &  \Gamma^2(\boldsymbol{\alpha}) \\
    \end{bmatrix}
    \in \mathbb{R}^{2N \times 2N} .
\end{equation}
The minimization problem subsequently becomes
\begin{equation} \label{eq:lagrangianmin_quad}
\begin{gathered}
    \min_{\boldsymbol{q}, \boldsymbol{w}, \boldsymbol{\alpha}} \boldsymbol{w}^{\intercal} \left(A(\boldsymbol{q}) + \lambda B(\boldsymbol{q}, \boldsymbol{\alpha}) \right) \boldsymbol{w} \\
    \text{s.t.} \quad \Gamma(\boldsymbol{\alpha}) \boldsymbol{f} = \Tilde{\boldsymbol{f}},
\end{gathered}
\end{equation}
where the objective function resembles the expectation value of the Hamiltonian with respect to $\boldsymbol{w}$. Therefore, the problem Hamiltonian is defined as
\begin{equation} \label{eq:problemhamiltonian}
    \mathcal{H}(\boldsymbol{q}, \boldsymbol{\alpha}) = A(\boldsymbol{q}) + \lambda B(\boldsymbol{q}, \boldsymbol{\alpha}) =
    \begin{bmatrix}
        K + \lambda K^2 & \lambda K \Gamma \\
        \lambda \Gamma K & \lambda \Gamma^2 \\
    \end{bmatrix}
\end{equation}
so that the minimum expected objective value is associated with the optimal configuration and displacements. It is noted that $\mathcal{H}(\boldsymbol{q}, \boldsymbol{\alpha})$ is positive semi-definite and symmetric, since the objective values in Eq. (\ref{eq:complianceminlagrangian}) are non-negative, and $K(\boldsymbol{q})$ and $\Gamma(\boldsymbol{\alpha})$ are both symmetric.

The compliance minimization problem in Eq. (\ref{eq:lagrangianmin_quad}) can be further reformulated  as
\begin{equation} \label{eq:lagrangianmin_unconstrained}
    \min_{\boldsymbol{q}, \boldsymbol{w}, \boldsymbol{\alpha}} \boldsymbol{w}^{\intercal} \mathcal{H}(\boldsymbol{q}, \boldsymbol{\alpha}) \boldsymbol{w} + \mu \Vert \Gamma(\boldsymbol{\alpha}) \boldsymbol{f} - \Tilde{\boldsymbol{f}} \Vert^2,
\end{equation}
where $\mu \in \mathbb{R}^+$ is another Lagrange multiplier.
The first term in Eq. (\ref{eq:lagrangianmin_unconstrained}) is computed on quantum computer. The second term is computed on classical computer, since it does not have a quadratic form and a corresponding Hamiltonian is not available. 
    
\subsection{Controlled problem Hamiltonian operator} \label{sec:problemhamiloniandecomposition}

The core operator is the problem Hamiltonian operator $U_C$. $U_C$ is applied to the second register that encodes $\boldsymbol{w}$. Since the problem Hamiltonian depends on the material configuration $\boldsymbol{q}$, $U_C$ is controlled by the first register.  The control scheme is derived based on the decomposition of $K(\boldsymbol{q})$ and the Trotter-Suzuki approximation \cite{suzuki1976generalized}.
Let 
\begin{equation} \label{eq:totalstiffnessmatrix}
    K(\boldsymbol{q}) = K_{0} + \sum_{j=1}^m q_j K_{j},
\end{equation}
where $K_{0}$ is the stiffness matrix of the base structure without optional elements, and $K_{j}$'s are stiffness matrices associated with the $m$ optional elements ($j = 1, \dots, m$). $q_j$ is a binary variable which indicates the presence or absence of the $j\textsuperscript{th}$ optional element.

For simplicity of notation, let
\begin{equation} \label{eq:hamil_noqubit}
\begin{split}
    Q^{(0)}(\boldsymbol{\alpha}) = 
    \begin{bmatrix}
        K_0 + \lambda K_0^2 & \lambda K_0 \Gamma(\boldsymbol{\alpha}) \\
        \lambda \Gamma(\boldsymbol{\alpha}) K_0 & \lambda \Gamma^2(\boldsymbol{\alpha}) \\
    \end{bmatrix}
\end{split},
\end{equation}
\begin{equation} \label{eq:hamil_onequbit}
\begin{split}
    & Q^{(1)}(q_j, \boldsymbol{\alpha}) = 
    q_j 
    \begin{bmatrix}
        K_j + \lambda K_0 K_j + \lambda K_j K_0 + \lambda K_j^2 & \lambda K_j \Gamma(\boldsymbol{\alpha}) \\
         \lambda \Gamma(\boldsymbol{\alpha}) K_j & \boldsymbol{0} \\
    \end{bmatrix},
\end{split}
\end{equation}
and
\begin{equation} \label{eq:hamil_twoqubit}
\begin{split}
    Q^{(2)}(q_j, q_k) = q_j q_k
    \begin{bmatrix}
        \lambda (K_j K_k + K_k K_j) & \boldsymbol{0} \\
        \boldsymbol{0} & \boldsymbol{0} \\
    \end{bmatrix}. \
\end{split}
\end{equation}
With the Trotter-Suzuki approximation, $U_C$ becomes 
\begin{equation} \label{eq:hamil_trotterized}
\begin{split}
    & U_C(\gamma, \boldsymbol{q}, \boldsymbol{\alpha}) \approx
    e^{-i\gamma Q^{(0)}(\cdot)}  \left[\prod_{j=1}^m e^{- i\gamma Q^{(1)}(\cdot)} \right]  \left[\prod_{j=1}^{m-1} \prod_{k=j+1}^m e^{- i\gamma Q^{(2)}(\cdot)} \right],
\end{split}
\end{equation}
where $\gamma$ is a rotation parameter.

The problem Hamiltonian operator in Eq. (\ref{eq:hamil_trotterized}) is a product of three terms. The first term, $\mathrm{exp}\left(-i\gamma Q^{(0)}(\boldsymbol{\alpha}) \right)$, is independent of $\boldsymbol{q}$. The second term, $\prod_{j=1}^m \mathrm{exp}\left({- i\gamma Q^{(1)}(q_j, \boldsymbol{\alpha})} \right)$, is a product of unitary operators, each of which is associated with one optional element $q_j$.  The third term, $\prod_{j=1}^{m-1} \prod_{k=j+1}^m \mathrm{exp}\left({- i\gamma Q^{(2)}(q_j, q_k)} \right)$, is also a product of unitary operators, and each one is associated with two optional elements $q_j$ and $q_k$. $U_C(\gamma, \boldsymbol{q}, \boldsymbol{\alpha})$ is a combination of operators corresponding to the present elements.

The application of $U_C(\gamma, \boldsymbol{q}, \boldsymbol{\alpha})$ on the second register is controlled by the first register.  The quantum state of the second register that encodes the displacements is entangled with the state of the first register which indicates the topology. As a result, the tasks of minimizing the compliance and solving the PDE constraints can be performed simultaneously.

\subsection{Gate complexity of problem Hamiltonian operator} \label{subsec:gatecomplexity}

The control problem Hamiltonian operator consists of one $Q^{(0)}$, $m$ of $Q^{(1)}$, and ${m(m-1)}/{2}$ of $Q^{(2)}$. Each Hermitian operator $Q^{(\cdot)}$ can be decomposed as
\begin{equation} \label{eq:paulidecomp}
    Q^{(\cdot)} = \sum_{j=1}^D E_j,
\end{equation} 
where $E_j$ is a one-sparse Hermitian matrix and $D$ is the number of $E_j$'s. The value of $D$ depends on the sparsity of $Q^{(\cdot)}$. The following theorem establishes an upper bound for $D$.
\begin{theorem} [\cite{kirby2021variational}]
    Suppose $Q^{(\cdot)}$ is a $d$-sparse Hermitian matrix. Then
    \begin{equation} \label{eq:Dupperbound}
        D \leq 4d^2 \left[\log_2 \left({{2}^{1/2}\Vert Q^{(\cdot)} \Vert_{\mathrm{max}}} \right)\right]
    \end{equation}
    where $\Vert Q^{(\cdot)} \Vert_{\mathrm{max}}$ is the maximum absolute value of any entry in $Q^{(\cdot)}$.
\end{theorem} 

The sparsity depends on each type of $Q^{(\cdot)}$. Let $v$ be the maximum number of elements connected to a single node in a truss structure. For $Q^{(0)}$, the largest value of $d$ is $4v$ since the base structure can contain elements which connect between one node and all other nodes. For $Q^{(1)}$, the largest value of $d$ is $2v + 4$ since $K_0$ only appears in the top-left submatrix of $Q^{(1)}$ and each $K_j$ has at most four non-zero entries in a row or column. For $Q^{(2)}$, the largest value of $d$ is $6$ since the two optional elements associated with $Q^{(2)}$ can share the same node.

With the Trotter-Suzuki approximation, the operator $e^{-i \gamma Q}$ is approximated as
\begin{equation}
    e^{-i \gamma Q} \approx \prod_{j = 1}^D e^{-i \gamma E_j}.
\end{equation} 
If $n$ is the number of qubits nneded to implement $e^{-i \gamma Q}$, then the number of single-qubit rotation and CNOT gates scales polynomially with $n$ \cite{wiebe2011simulating}. 

$U_C$ consists of $\mathcal{O}\left(v^2 \log_2 \left({2}^{1/2} \Vert Q^{(0)} \Vert_{\mathrm{max}} \right) \mathrm{poly}(n)\right)$ gates for $Q^{(0)}$, $\mathcal{O}\left(v^2 \log_2 \left({2}^{1/2} \Vert Q^{(1)} \Vert_{\mathrm{max}} \right)\mathrm{poly}(n)\right)$  gates for each $Q^{(1)}$, and $\mathcal{O}\left(\log_2 \left({2}^{1/2} \Vert Q^{(2)} \Vert_{\mathrm{max}} \right)\mathrm{poly}(n)\right)$ gates for each $Q^{(2)}$. The gate complexity of $U_C$ depends on the arrangement of elements in the structure.  As the number of elements connected to a single node decreases, the sparsity of the stiffness matrices increases. As a result, fewer quantum gates are needed to construct the control scheme. 

\subsection{Quantum circuit architecture} \label{sec:quantumcircuit}

The overall quantum circuit for the proposed VQA is shown in Figure \ref{fig:toqaboa_circuit}. In the first register $|\boldsymbol{q}\rangle = \otimes_{j=1}^{m} |q_j\rangle$, the computational basis encodes a total of $2^m$ material configurations. In the second register $|\boldsymbol{w}\rangle = \otimes_{j=1}^{n} |w_j\rangle$, the amplitudes of $n = \lceil \log_2 N \rceil +1$ qubits represent the estimated displacements $\boldsymbol{u}$ and loads $\boldsymbol{f}$ collectively. All qubits are initialized as $|0\rangle$.

\begin{figure*}[htbp!]
    \centering
    \includegraphics[width=\textwidth, trim={0in 1in 0 0.5in}]{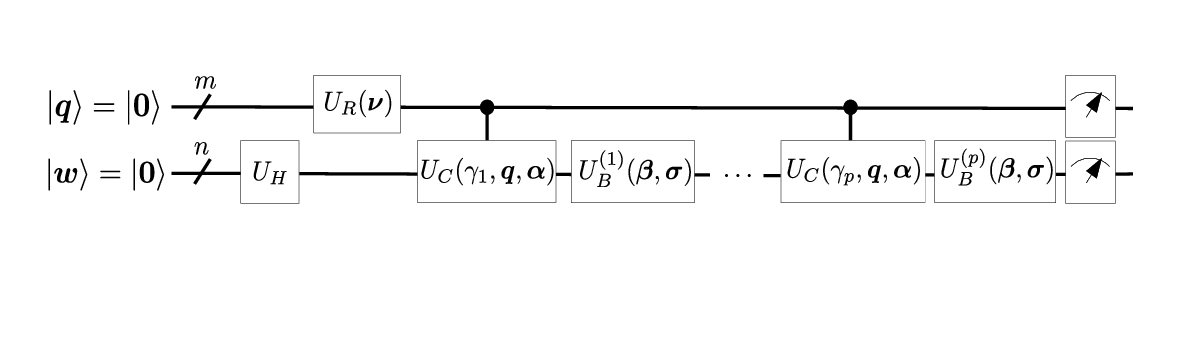}
    \caption{The architecture of the proposed variational quantum topology optimization algorithm}
    \label{fig:toqaboa_circuit}
    \vspace{5pt}
\end{figure*}

The overall operator is defined as
\begin{equation}
    W=\left[(I\otimes U_B)W_C\right]^p (U_R \otimes U_H),
\end{equation}
where $p$ is the number of repetitions. $U_H$ is the Walsh-Hadamard transform, defined as
\begin{equation}
    U_H=H^{\otimes n}.
\end{equation}
$U_R$ is defined as
\begin{equation} \label{eq:UR}
    U_R(\boldsymbol{\nu}) = \otimes_{j=1}^m e^{-i \nu_j \sigma_X},
\end{equation}
where $\boldsymbol{\nu} = \begin{bmatrix} \nu_1 & \dots & \nu_m\\ \end{bmatrix}^\intercal$, $\nu_j$'s are continuous parameters, and $\sigma_X$ is the Pauli-X matrix. $W_C$ is defined as
\begin{equation}
    W_C(\gamma, \boldsymbol{q}, \boldsymbol{\alpha}) = \sum_{\boldsymbol{q} \in \{0, 1\}^{m}} |\boldsymbol{q}\rangle \langle \boldsymbol{q}| \otimes U_C(\gamma, \boldsymbol{q}, \boldsymbol{\alpha}).
\end{equation}
$U_B$ is defined as
\begin{equation}
\begin{split}
    U_B(\boldsymbol{\beta}, \boldsymbol{\sigma}) = \left(\otimes_{k=1}^n e^{-i\beta_{k}\sigma_X}\right) \left(\otimes_{k=1}^n |0\rangle\langle 0|I + |1\rangle\langle 1| e^{-i\beta_{n+k}\sigma_k} \right) \left(\otimes_{k=1}^n e^{-i\beta_{2n+k}\sigma_Y}\right),
\end{split}
\end{equation}
where $\boldsymbol{\beta} = \begin{bmatrix} \beta_1 & \dots & \beta_{3n}\\ \end{bmatrix}^\intercal$, $\beta_k$'s are continuous circuit  parameters,  $\boldsymbol{\sigma} = \begin{bmatrix} \sigma_1 & \dots & \sigma_n\\ \end{bmatrix}^\intercal$, and $\sigma_k \in \{I, \sigma_X, \sigma_Y, \sigma_Z\}$. Here, both continuous circuit parameters and the circuit architecture as discrete choices of $\sigma_k$'s are optimized. 

The application of $U_R$ on $|\boldsymbol{q}\rangle$ results in the exploration of all material configurations. For constrained TO problems, a method of encoding the volume and connectivity constraints is needed to restrict the search to feasible configurations. 

\section{Constraint Encoding} \label{sec:constraintencoding}
To solve constrained TO problems, a sequence of controlled and uncontrolled rotation gates $\Tilde{U}_R$ are applied on $|\boldsymbol{q}\rangle$ instead of $U_R$ described in the previous section. This sequence is defined as
\begin{equation} \label{eq:constrainedrotation}
    \Tilde{U}_R(\xi, \boldsymbol{\nu}) = U_{T}(\xi)U_{R}(\boldsymbol{\nu}),
\end{equation}
where $\xi \in \mathbb{R}^-$ is a rotation parameter and $U_{T}$ is a sequence of controlled rotation gates. Each rotation gate is applied to a qubit about the x-axis in the Bloch sphere. 
By applying the rotation gates and restricting the ranges of parameters, the search space is reduced to feasible configurations which satisfy the constraints.

Two types of constraints are encoded. The first type is a connectivity constraint, which ensures that elements in the structure are connected with each other. A configuration which satisfies connectivity constraints does not contain the checkerboard patterns or elements which are located in the centers of voids. The connectivity constraints are encoded with the controlled rotation gates. For all controlled rotation gates, $\xi$ is set to $-\nu^*$, where $\nu^*$ is the upper bound of $\nu_j$. Setting $\xi$ to this value is needed to reduce the probabilities of measuring infeasible configurations. To illustrate with an example, let $q_j$ and $q_k$ denote the presence of two adjacent optional elements. When $q_k = 1$, it is required that $q_j = 1$ so that the optional elements are connected to the base structure. A controlled rotation gate is applied where $|q_j\rangle$ is the control qubit and $|q_k\rangle$ is the target qubit. When $|q_j\rangle = |0\rangle$, the gate decreases the amplitude of $|1\rangle$ in $|q_k\rangle$.  As a result, the probability of measuring configurations where $q_j = 0$ and $q_k = 1$ decreases. Figure \ref{fig:MBB_QC} shows an example of $\Tilde{U}_R$, which encodes the constraint that $q_4 = 0$ when both $q_2 = 0$ and $q_3 = 0$. In this case, the rotation gate applied on $|q_4\rangle$ is controlled by both $|q_2\rangle$ and $|q_3\rangle$. The rotation is only performed when both $|q_2\rangle = |0\rangle$ and $|q_3\rangle = |0\rangle$.

\begin{figure}[htbp!]
    \centering
    \includegraphics[width=0.4\textwidth, trim={0.25in 0.5in 0.25in 0.3in}, clip]{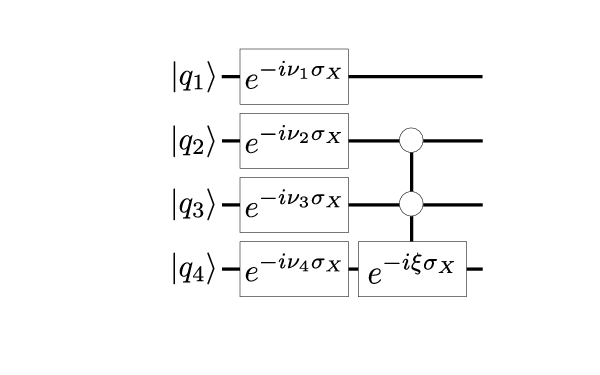}
    \caption{An example of connectivity constraints encoded in $\Tilde{U}_R$}
    \label{fig:MBB_QC}
\end{figure}

The second type is a volume constraint, which limits the maximum number of included elements. The volume constraints are encoded by setting $\nu^{*}$.  By enforcing this constraint, some of the qubits are likely to be measured as $|1\rangle$, whereas the other qubits are likely to be measured as $|0\rangle$. In general, as $\nu^{*}$ increases, the probability of including more elements in the optimal configuration increases.  The following theorem establishes the guidance of choosing the maximum value $\nu^{*}$.
\begin{theorem} \label{thm:parametervolumeconstraint}
    Let $m \in \mathbb{Z}^+$ denote the total number of optional elements and $s \in \mathbb{Z}^+$ be the maximum number of allowable ones. The probability that the number of optional elements $q^*$ in the obtained optimal structure is more than $s$ has an upper bound, as 
    \begin{equation} \label{eq:maxnubound}
        \mathbb{P}(q^*>s) \leq \sum_{l = s+1}^m {m \choose l} \sin^{2l} \left(\frac{\nu^{*}}{2} \right)\cos^{2(m-l)} \left(\frac{\nu^{*}}{2} \right).
    \end{equation}
\end{theorem}
\begin{proof}
    The initial quantum state is $|\boldsymbol{0}\rangle^{\otimes m}$. Applying $U_R$ results in the transformed state
    \begin{equation}
        |\psi\rangle = \sum_{\boldsymbol{q} \in \{0, 1\}^m} \left[\prod_{j=1}^m a \left(q_j, {\nu}_j \right) \right] |\boldsymbol{q}\rangle,
    \end{equation}
    where
    \begin{equation}
        a(q_j, {\nu}_j) = \begin{cases}
            -i \sin \left( \frac{{\nu}_j}{2} \right) & \text{if $q_j = 1$} \\
            \cos \left( \frac{{\nu}_j}{2} \right) & \text{if $q_j = 0$} \\
        \end{cases}.
    \end{equation}
    The probability of obtaining any configuration with $l$ optional elements is then 
    \begin{equation} \label{eq:prob_joptelements}
        \mathbb{P}(q^*=l)  = \sum_{M_l} \left[ \prod_{j \in M_l} \sin^2 \left(\frac{{\nu}_j}{2} \right) \right]\left[ \prod_{k \notin M_l} \cos^2 \left(\frac{{\nu}_k}{2} \right) \right]  \quad \left(l=s+1,\ldots,m\right),
    \end{equation}
    where $M_l \subseteq \{1, \dots, m\}$ such that $|M_l| = l$. The right-hand side of Eq. (\ref{eq:prob_joptelements}) is a sum over all possible subsets with $l$ integers.
    
    As ${\nu}_j$ increases from 0 to $\pi$, $\sin({\nu}_j/2)$ increases monotonically.  Since $\nu^{*}$ is the upper bound for all ${\nu}_j$'s,
    \begin{equation}
        \sum_{M_l} \left[ \prod_{j \in M_l} \sin^2 \left(\frac{{\nu}_j}{2} \right) \right]\left[ \prod_{k \notin M_l} \cos^2 \left(\frac{{\nu}_k}{2} \right) \right] \leq {m \choose l} \sin^{2l} \left(\frac{\nu^{*}}{2} \right)\cos^{2(m-l)} \left(\frac{\nu^{*}}{2} \right),
    \end{equation}
    where
$       {m \choose l} = \frac{m!}
        {l!(m-l)!}
$.
    Therefore,
    \begin{equation}
       \mathbb{P}(q^*=l) \leq {m \choose l} \sin^{2l} \left(\frac{\nu^{*}}{2} \right)\cos^{2(m-l)} \left(\frac{\nu^{*}}{2} \right)
    \end{equation}
and
    \begin{equation}
         \mathbb{P}(q^*>s) = \sum_{l = s+1}^m \mathbb{P}(q^*=l) \leq \sum_{l = s+1}^m {m \choose l} \sin^{2l} \left(\frac{\nu^{*}}{2} \right)\cos^{2(m-l)} \left(\frac{\nu^{*}}{2} \right).
    \end{equation}
\end{proof}

Theorem \ref{thm:parametervolumeconstraint} provides a guidance to set the value of $\nu^{*}$. Given a probability upper bound of constraint violation, the value of $\nu^{*}$ can be obtained by solving the nonlinear equation in Eq. (\ref{eq:maxnubound}). Some examples are illustrated in Figure \ref{fig:probupperboundcurves}, where each curve represents the solutions of the nonlinear equation. The value of $\nu^*$ can be identified for the given values of $m$ and $s$. For a fixed value of $m$, the upper probability decreases as $s$ increases. It is also observed that the curve corresponding to $s=m$ is constantly zero since there is no volume constraint involved.

\begin{figure*}[ht!] 
\centering
\includegraphics[width=0.5\linewidth, trim=0 0 0 0, clip]{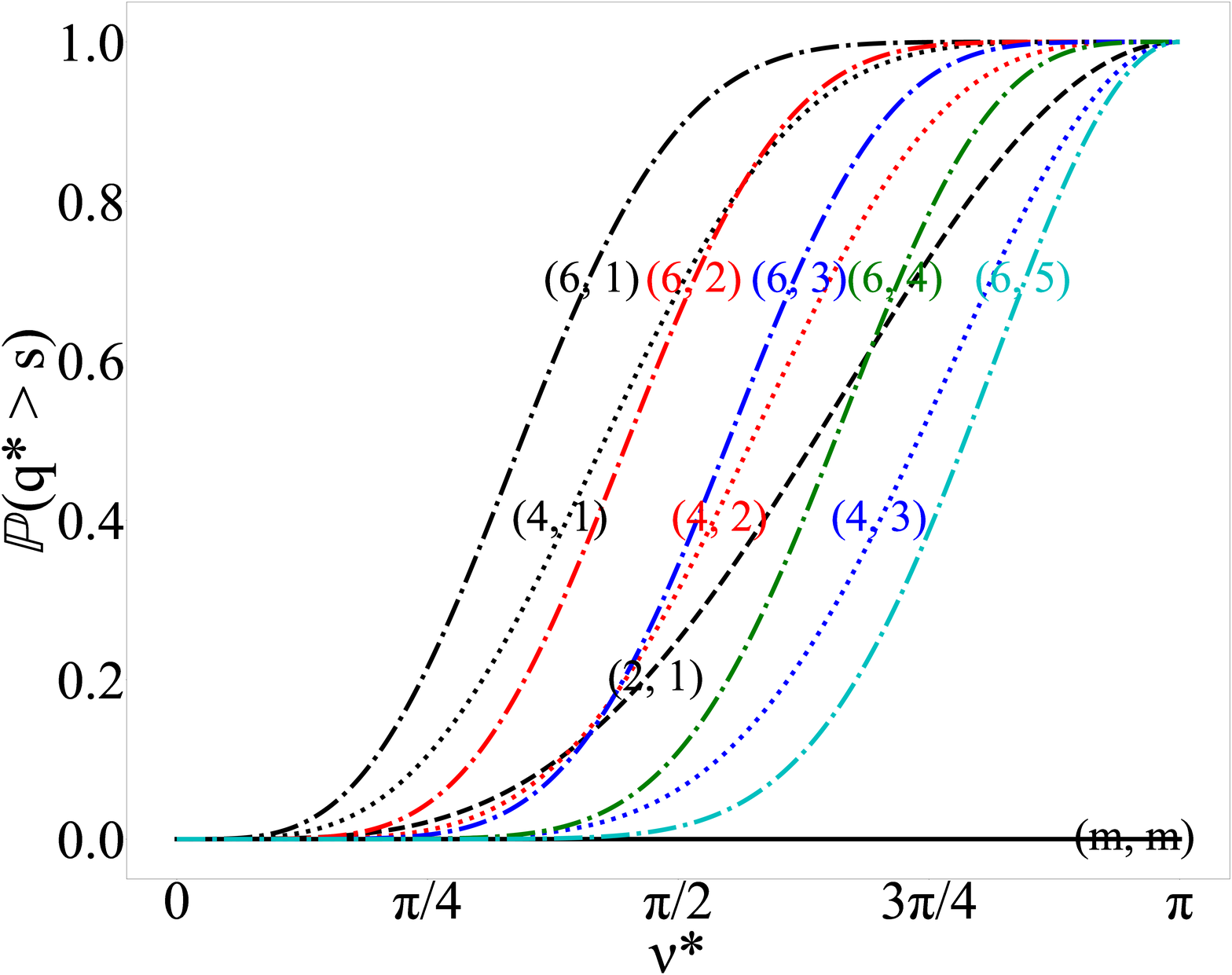}
\caption{The upper bounds of constraint violation probability for different values of ($m$, $s$). The dashed line corresponds to $m=2$, the dotted lines correspond to $m=4$, the dash-dotted lines correspond to $m=6$, and the solid line corresponds to $s=m$.}
\label{fig:probupperboundcurves}
\end{figure*}

\section{Experimental Results} \label{sec:results}

\subsection{Compliance minimization of truss structures} \label{subsec:trussexample}

The proposed VQA for TO is first demonstrated with three two-dimensional truss structures. The three structures are shown in Figure \ref{fig:trussstructures}, where the $x$ and $y$ axes represent spatial coordinates in inches. In Figure \ref{subfig:trussstructure1}, truss structure \#1 contains three pin supports and a vertical load of 8,000 kilopounds (kips) is applied at node 2. In Figure \ref{subfig:trussstructure2}, truss structure \#2 contains two pin supports and a vertical load of 8,000 kips is applied at node 2. In Figure \ref{subfig:trussstructure3}, truss structure \#3 contains two pin supports and a vertical load of 5,500 kips is applied at node 5. Each truss consists of a base structure and optional elements. The base structures are illustrated as solid lines, whereas the optional elements are depicted as dashed lines. In all truss structures, the cross-sectional area and elastic modulus of each element are 8.5 in\textsuperscript{2} and 29,000 ksi, respectively.

\begin{figure*}[ht!]
\centering
\begin{subfigure}{.32\textwidth}
  \centering
  \includegraphics[width=\linewidth, trim=0 0.6in 0 0, clip]{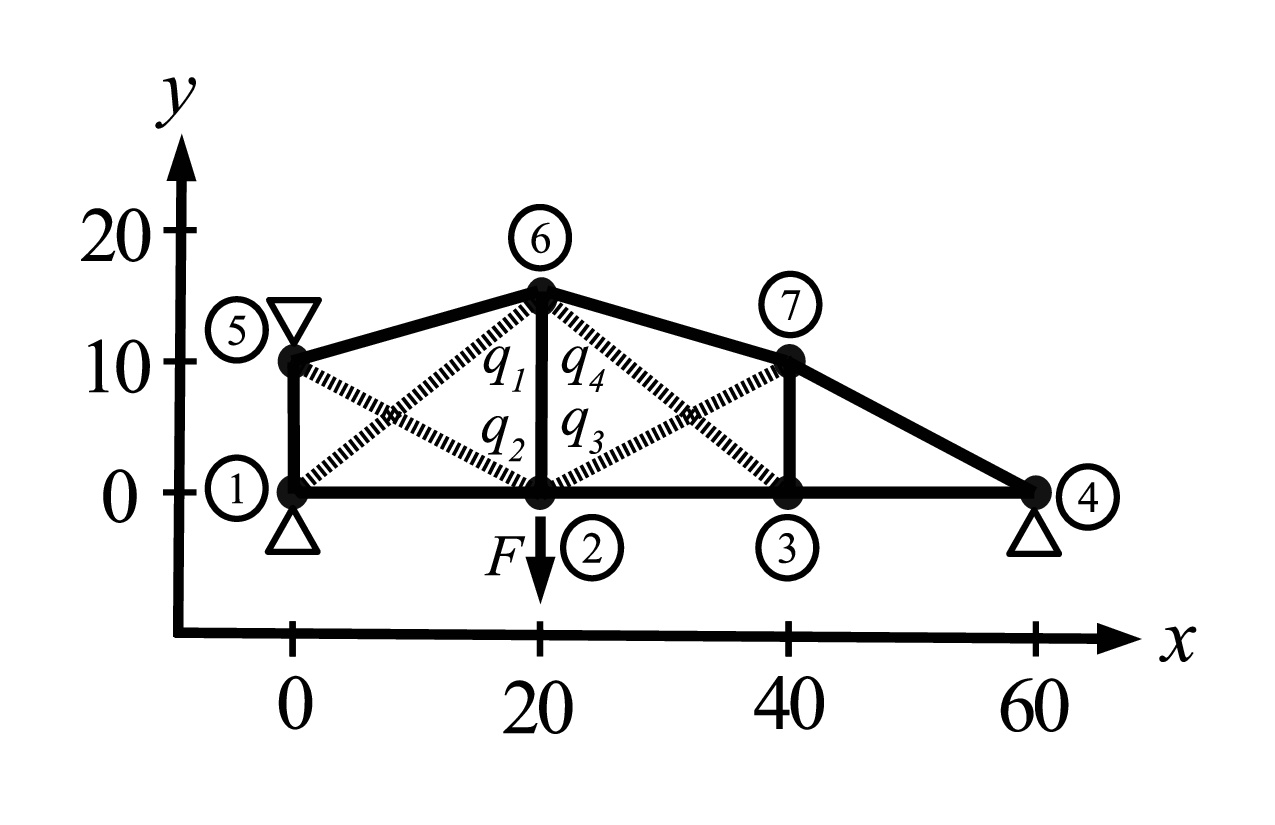}
  \caption{}
  \label{subfig:trussstructure1}
\end{subfigure}%
\begin{subfigure}{.32\textwidth}
  \centering
  \includegraphics[width=\linewidth, trim=0 0.6in 0 0, clip]{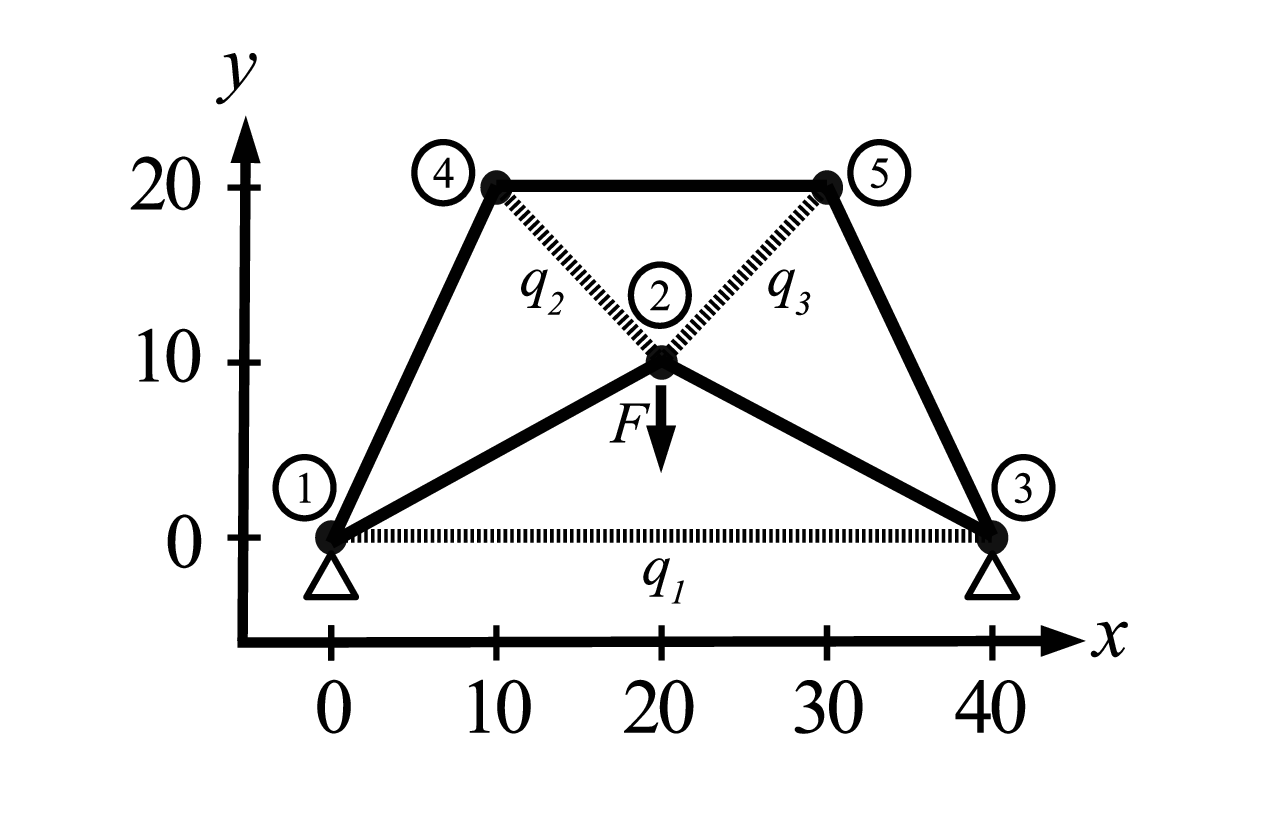}
  \caption{}
  \label{subfig:trussstructure2}
\end{subfigure}
\begin{subfigure}{.32\textwidth}
  \centering
  \includegraphics[width=\linewidth,trim=0 0.6in 0 0, clip]{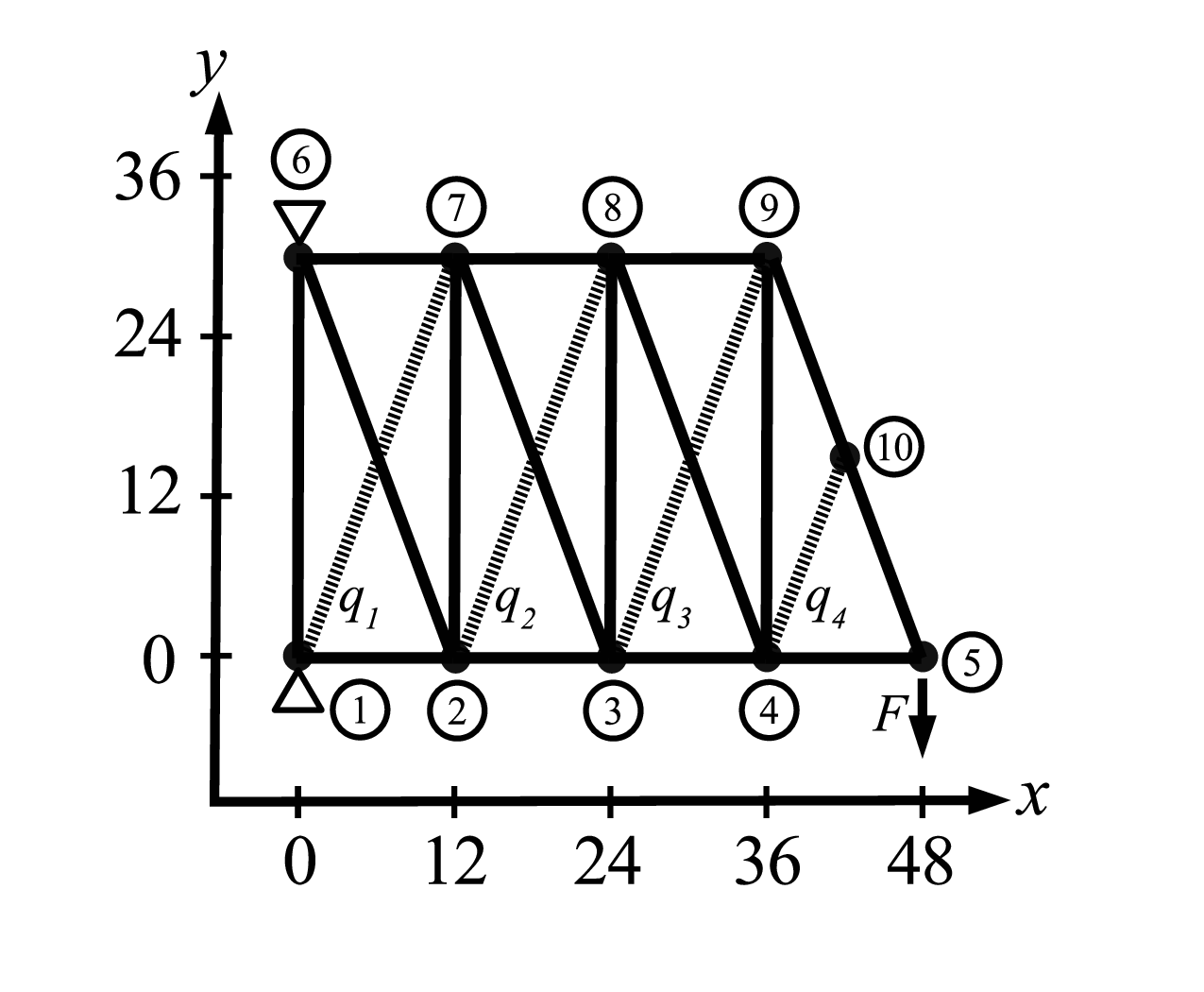}
  \caption{}
  \label{subfig:trussstructure3}
\end{subfigure}
\caption{Illustrations of truss structures (a) \#1, (b) \#2, and (c) \#3}
\label{fig:trussstructures}
\end{figure*}

Bayesian optimization is performed to find the parameters and architecture of quantum circuit simultaneously. The surrogate model, which is a Gaussian process regression model with the Mat\'ern covariance kernel, generates a probabilistic distribution of objective functions. The upper confidence bound is used as the acquisition function to search for the optimal parameters. 

For each truss structure, the proposed VQA is performed for ten runs to obtain the distributions of the optimal configurations. For truss structures \#1 and \#2,  50 iterations of Bayesian optimization are done for each run.  $\lambda$ is initially 0.1 and increases by 0.005 after each iteration.  For truss structure \#3, 100 iterations are done for each run.  $\lambda$ is initially 1.0 and increases by 0.05 after each iteration.  The Lagrange multiplier $\mu$ is set to a fixed value of 15 for all three structures. The circuit depth is set to $p = 4$. 

The maximum numbers of elements allowed to be included are two for truss structures \#1 and \#3, whereas the maximum number of elements is one for truss structure \#2. To assess the sensitivity of setting the probabilities for violating volume constraints and verify the theoretical upper bound in Theorem \ref{thm:parametervolumeconstraint}, the probabilities of violating constraints are also estimated for each structure. Sampling procedures are taken, where eleven values of $\nu^*$, which are multiples of 0.1$\pi$ ranging from 0 to $\pi$, are selected. 100,000 shots are taken for each value on the quantum circuit in Figure \ref{fig:toqaboa_circuit} without optimizing the parameters. By counting the number of times $N_f$ that a configuration with more than $s$ optional elements are generated out of the 100,000 shots, the constraint violation probability is estimated as $\mathbb{P}(q^* > s) \approx {N_f}/{100000}$.

\begin{figure*}[ht!] 
\centering
\begin{subfigure}{.45\textwidth}
  \centering
  \includegraphics[width=\linewidth, trim=0 0 -1.25in 0, clip]{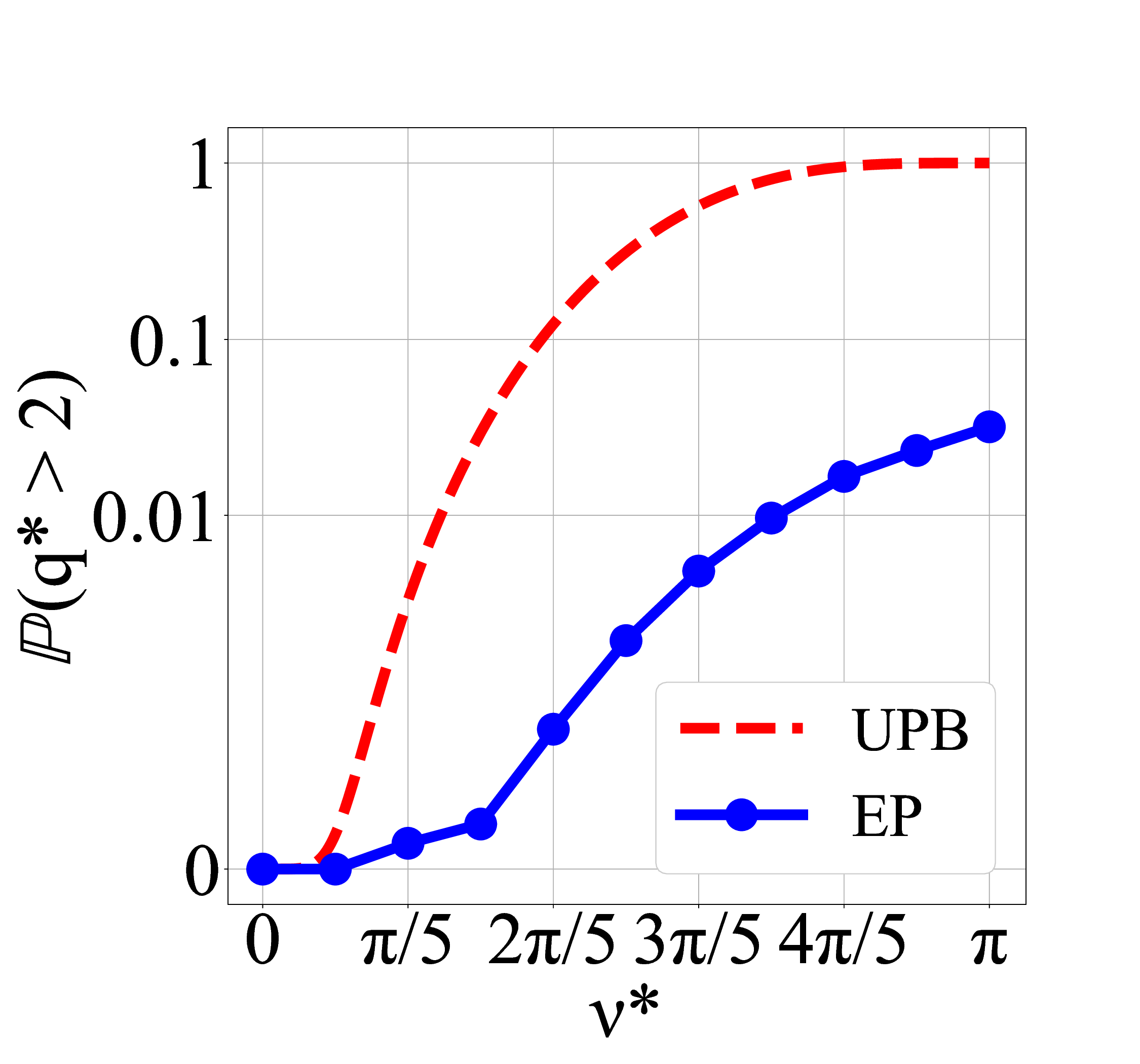}
  \caption{}
  \label{subfig:TrussExample2_LogProbCurves}
\end{subfigure}%
\begin{subfigure}{.45\textwidth}
  \centering
  \includegraphics[width=\linewidth, trim=0 0 -1.25in 0, clip]{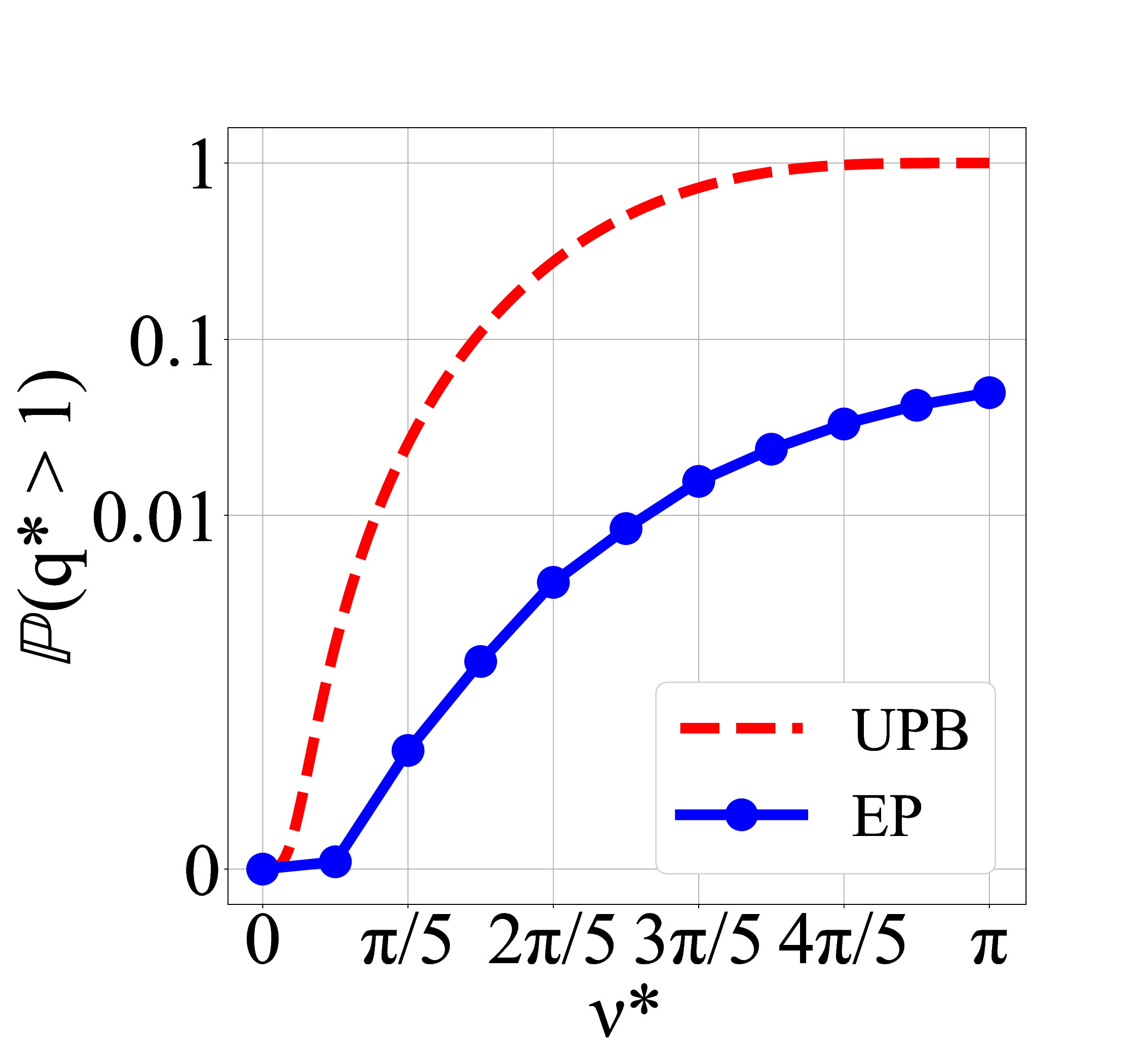}
  \caption{}
  \label{subfig:TrussExample3_LogProbCurves}
\end{subfigure}
\caption{Comparison between upper probability bounds (UPBs) and estimated probabilities (EPs) of violating volume constraints for truss structures (a) \#1 and \#3, and (b) \#2}
\label{fig:probcurves_trussstructures}
\end{figure*}

The estimated constraint violation probabilities and the upper probability bounds from Eq. (\ref{eq:maxnubound}) are shown in Figure \ref{fig:probcurves_trussstructures} for the three truss structures.  In both plots, the upper probability bounds and estimated probabilities increase as $\nu^*$ increases from 0 to $\pi$. This implies that the volume constraints are not as strictly enforced when $\nu^*$ is large. For all values of $\nu^*$, the probability of volume constraint violation is smaller than the theoretical upper bound. 
Without sampling, the theoretical probability upper bounds provide the guidance to choose the values of $\nu^*$ in setting the volume constraints. In this example, we set the upper bound of constraint violation probability to be $0.5$. Accordingly, $\nu^*$ is set to 1.801 for truss structures \#1 and \#3, whereas $\nu^*$ is set to 1.571 for truss structure \#2. Notice that although the upper bound of probability is set to be 0.5, the actual probabilities of constraint violation are small. For truss structures \#1 and \#3, the actual probability is 0.004. For truss structure \#2, the actual probability is 0.0084.

\begin{figure*}[t!]
\centering
\begin{subfigure}{0.35\textwidth}
  \centering 
  \includegraphics[width=\linewidth, trim=0 0 0 0, clip]{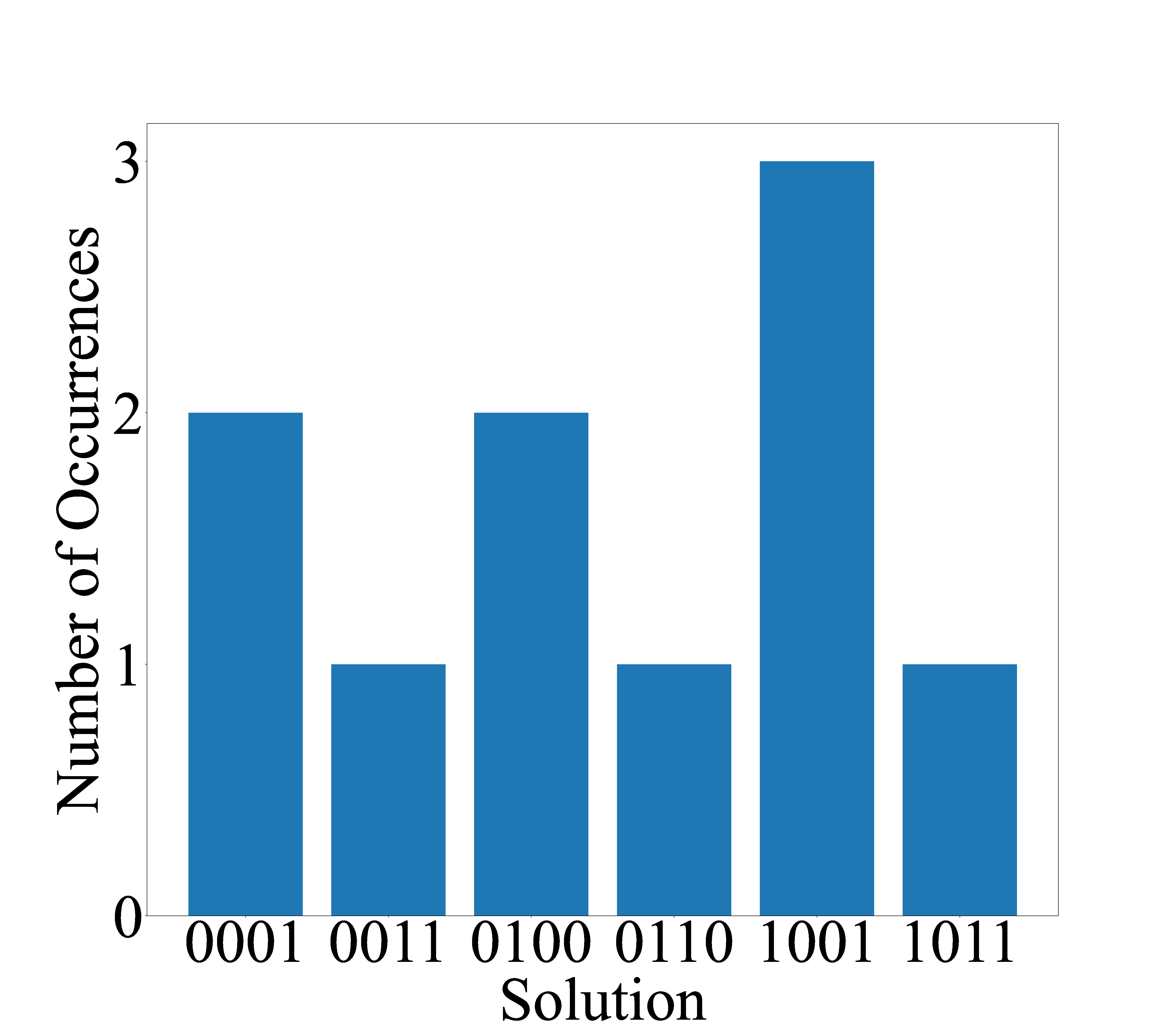}
  \caption{}
  \label{subfig:optimalsols_trussstructure1_maxprob001}
\end{subfigure}
\begin{subfigure}{0.22\textwidth}
  \centering
    \includegraphics[width=\linewidth, trim=0 0 0 0, clip]{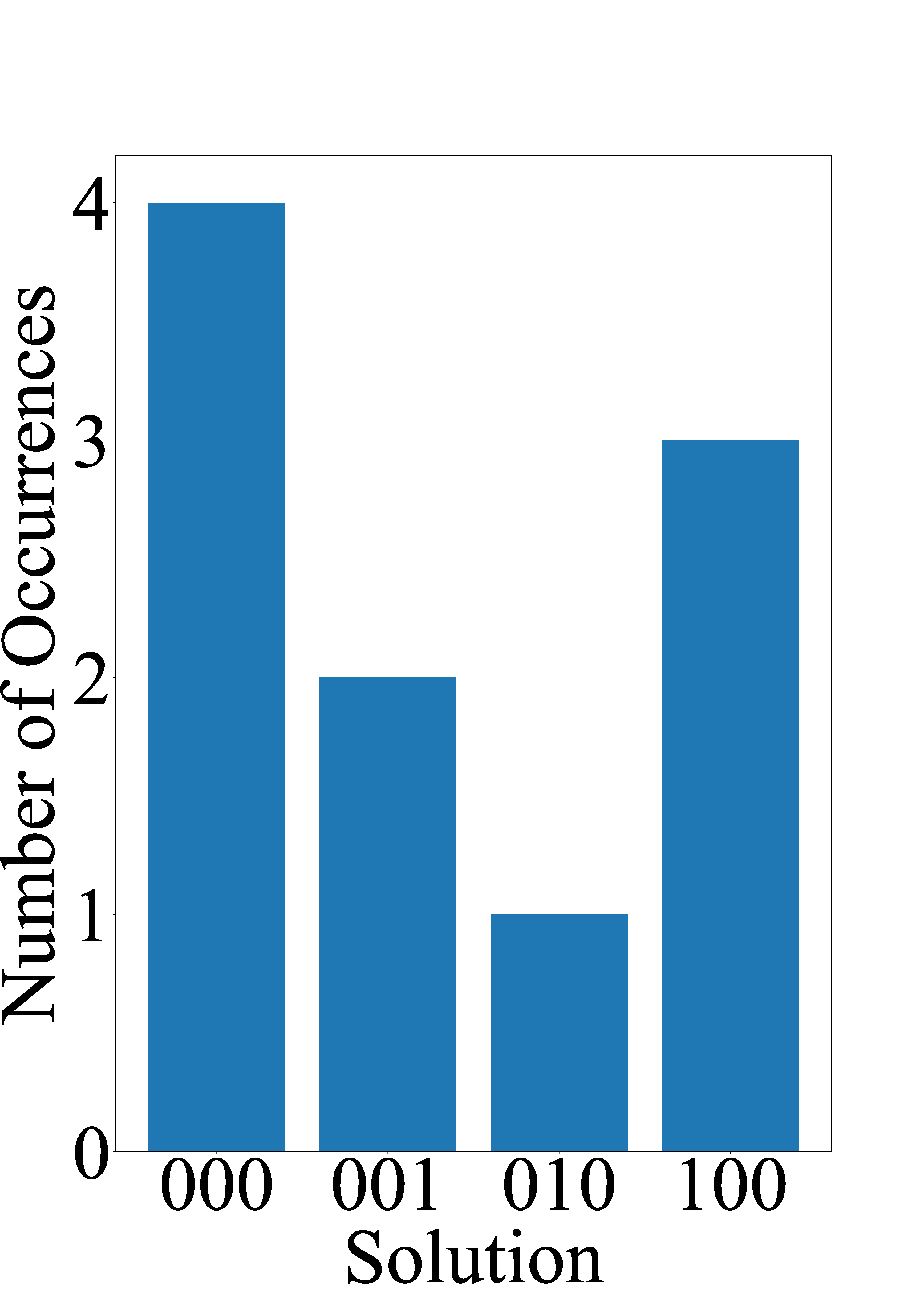}
  \caption{}
  \label{subfig:optimalsols_trussstructure2_maxprob001}
\end{subfigure}
\begin{subfigure}{0.4\textwidth}
  \centering
  \includegraphics[width=\linewidth, trim=0in 0 0in 0, clip]{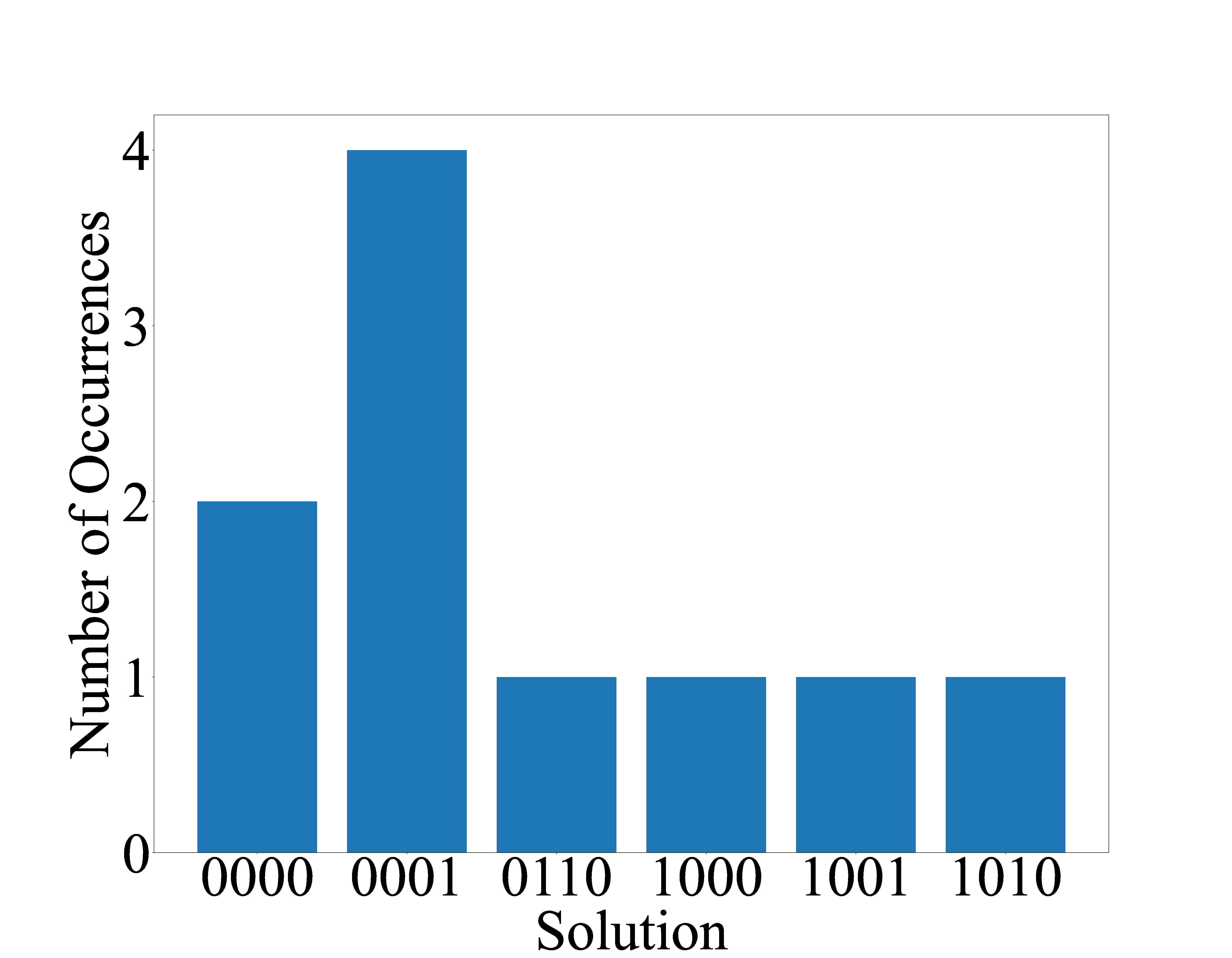}
  \caption{}
  \label{subfig:optimalsols_trussstructure3_maxprob001}
\end{subfigure}
\caption{Distributions of the optimal configurations for truss structures (a) \#1, (b) \#2, and (c) \#3 with volume constraints}
\label{fig:trussstructures_optimalsols_volconstraint}
\end{figure*}

The VQA is performed with the noiseless simulator from Qiskit. The distributions of optimal configurations are shown in Figure \ref{fig:trussstructures_optimalsols_volconstraint}. The VQA is able to find the true optimal configurations for all three structures with success probabilities ranging from 10\% to 30\%. The true optimal configurations are $q_1q_2q_3q_4 = 1001$ for truss structure \#1, $q_1q_2q_3 = 001$ and $q_1q_2q_3 = 010$ for truss structure \#2, and $q_1q_2q_3q_4 = 0110$ for truss structure 
\#3. In Figure \ref{subfig:optimalsols_trussstructure1_maxprob001}, three of the ten runs result in $q_1q_2q_3q_4 = 1001$. In Figure \ref{subfig:optimalsols_trussstructure2_maxprob001}, two runs result in $q_1q_2q_3 = 001$, whereas one run results in $q_1q_2q_3 = 010$. In Figure \ref{subfig:optimalsols_trussstructure3_maxprob001}, one run results in $q_1q_2q_3q_4 = 0110$. In addition, among all three distributions of optimal configurations, only one run results in a configuration which violates the volume constraint. This  run results in $q_1q_2q_3q_4 = 1011$ for truss structure \#1, which contains three optional elements.

\subsection{Compliance minimization of MBB beams} \label{subsec:beamexample}

\begin{figure*}[ht!]
\centering
\begin{subfigure}{.45\textwidth}
  \centering
  \includegraphics[width=0.9\linewidth, trim=1.6in 1in 1.6in 0.8in, clip]{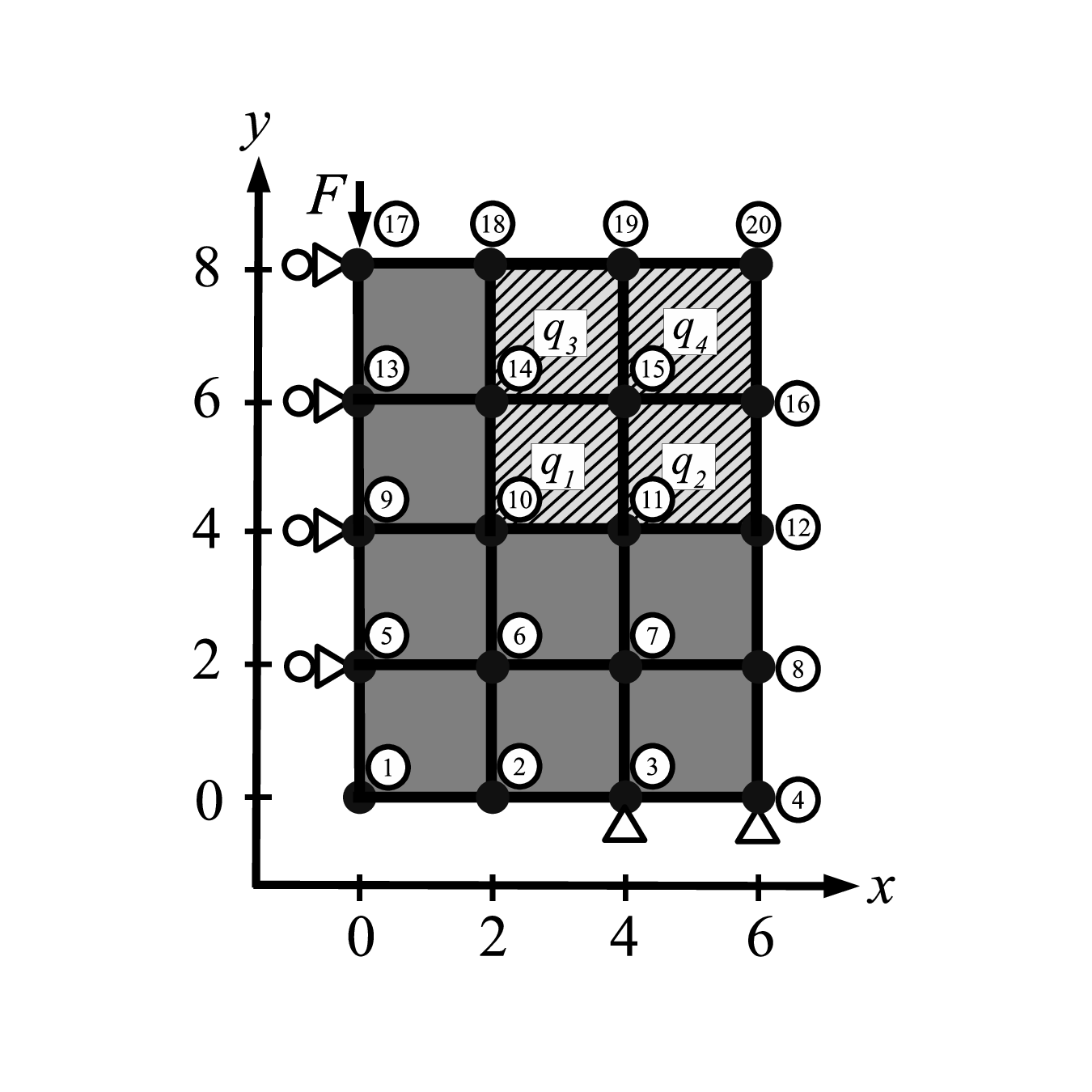}
  \caption{}
  \label{subfig:MBBstructure1}
\end{subfigure}%
\begin{subfigure}{.55\textwidth}
  \centering
  \includegraphics[width=\linewidth, trim=1.25in 1in 1in 1.6in, clip]{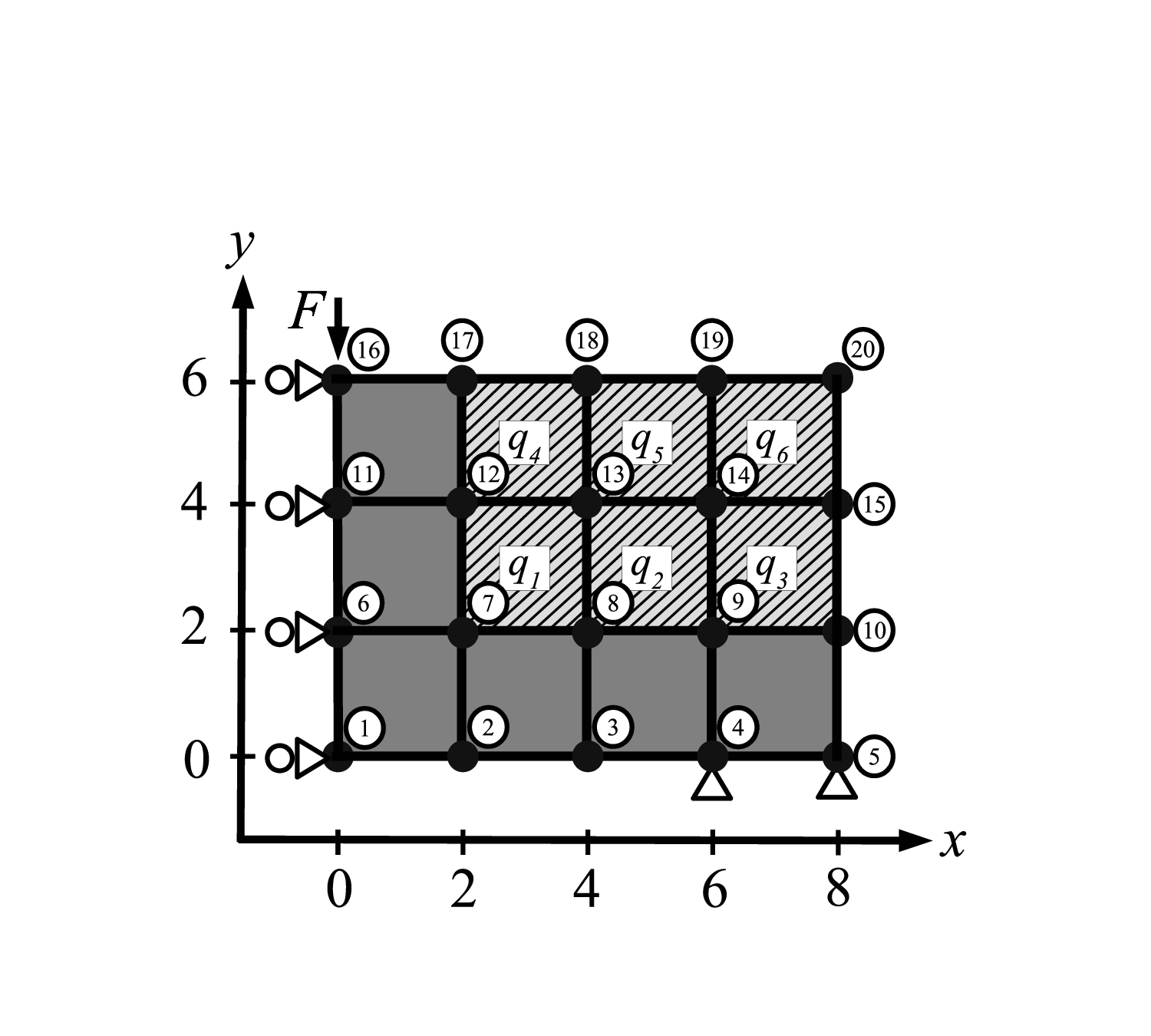}
  \caption{}
  \label{subfig:MBBstructure2}
\end{subfigure}
\caption{Illustrations of MBB beams (a) \#1 and (b) \#2}
\label{fig:MBBbeamstructures}
\end{figure*}

As the second example, the problems of compliance minimization for two MBB beams in Figure \ref{fig:MBBbeamstructures} are solved.  In Figure \ref{subfig:MBBstructure1}, beam \#1 consists of twelve square elements arranged as a $4\times 3$  rectangular grid. The side length of each element is 2 in. Eight elements, which are depicted as solid gray regions, comprise the base structure. The other four elements, which are depicted as striped regions, are optional elements. Four roller supports and two pin supports are located along the left and bottom edges, respectively. A vertical load of 200,000 kips is applied at node 17. In Figure \ref{subfig:MBBstructure2}, the twelve square elements in beam \#2, each with a side length of 2 in, are arranged as a $3 \times 4$  rectangular grid. Six elements comprise the base structure, whereas the other six elements are optional. Similar to beam \#1, four roller supports and two pin supports are located along the left and bottom edges. A vertical load of 200,000 kips is applied at node 16.

For beam \#1, one volume constraint and one connectivity constraint are considered. For the volume constraint, the maximum number of optional elements allowed to be included is two. For the connectivity constraint, $q_4$ must be zero if $q_2$ and $q_3$ are both zeros.  The implementation of this connectivity constraint is shown in Figure \ref{fig:MBB_QC}. Beam \#2 involves one volume constraint and two connectivity constraints. For the volume constraint, the maximum number of optional elements is two. For the first connectivity constraint, $q_5$ must be zero if $q_2$ and $q_4$ are both zeros. For the second connectivity constraint, $q_6$ must be zero if $q_3$ and $q_5$ are both zeros.

\begin{figure*}[ht!] 
\centering
\begin{subfigure}{.5\textwidth}
  \centering
  \includegraphics[width=\linewidth, trim=0 0 -1.25in 0, clip]{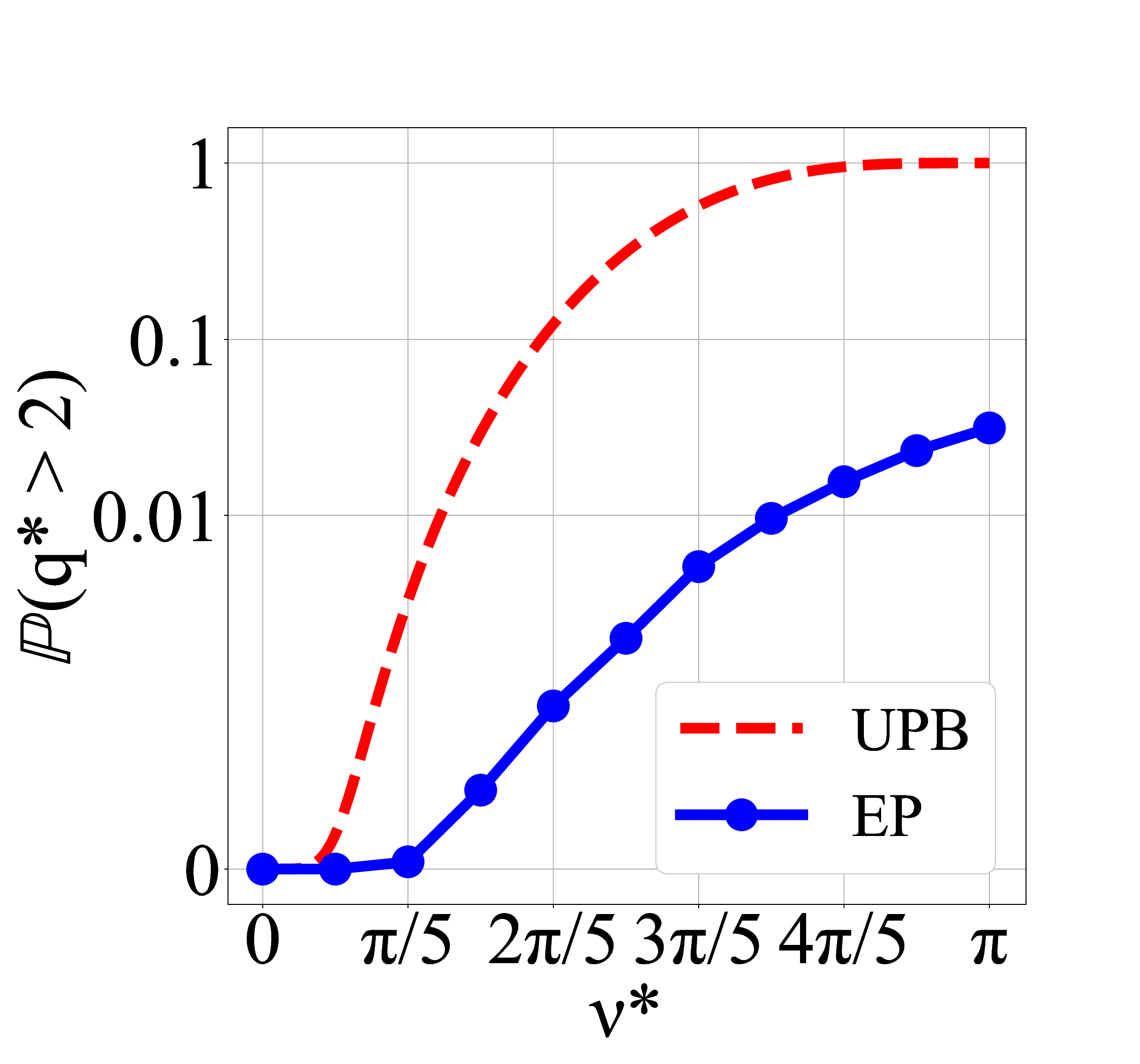}
  \caption{}
  \label{subfig:MBBExample1_LogProbCurves}
\end{subfigure}%
\begin{subfigure}{.5\textwidth}
  \centering
  \includegraphics[width=\linewidth, trim=0 0 -1.25in 0, clip]{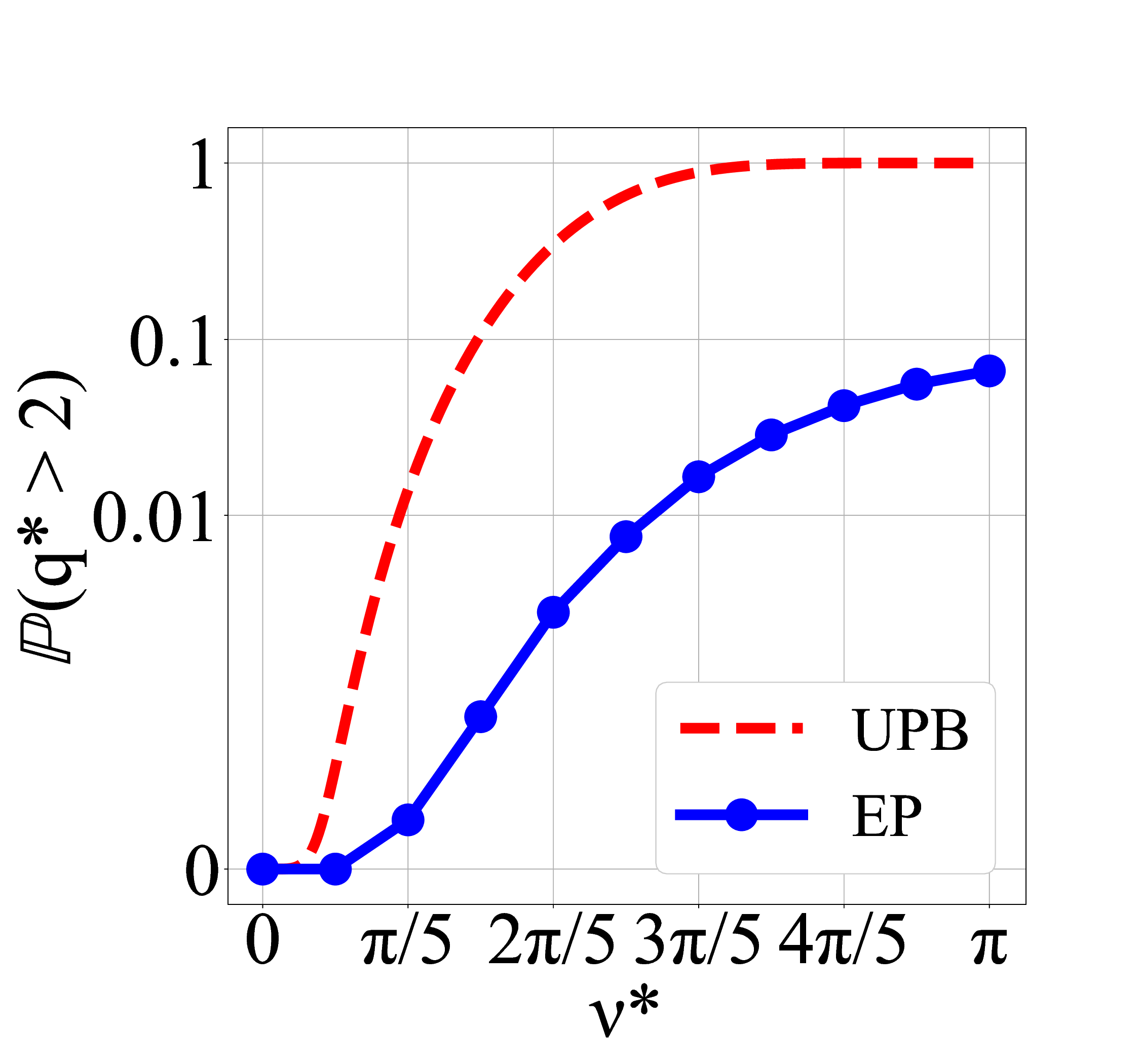}
  \caption{}
  \label{subfig:MBBExample2_LogProbCurves}
\end{subfigure}
\caption{Comparison between upper probability bounds (UPBs) and estimated  probabilities (EPs) of violating volume constraints for MBB beams (a) \#1 and (b) \#2}
\label{fig:probcurves_MBBstructures}
\end{figure*}

Similar to the previous truss problems, the upper bound of volume constraint violation is set to be $0.5$ for the MBB problems. For the purpose of illustration, the probabilities of violation and the upper bounds are similarly compared in Figure \ref{fig:probcurves_MBBstructures}.  
Based on the probability upper bound of $0.5$, the values of parameter upper bounds $\nu^*$'s are set to be $1.571$ for MBB beam \#1 and $1.413$ for beam \#2.  Accordingly, the values of $\xi$ in the quantum circuit are $-1.571$ for beam \#1 and $-1.413$ for beam \#2 to reduce the probability of measuring infeasible configurations.

\begin{figure*}[t!]
\centering
\begin{subfigure}{0.42\textwidth}
  \centering 
  \includegraphics[width=\linewidth, trim=5.5in 7in -3in 0in, clip]{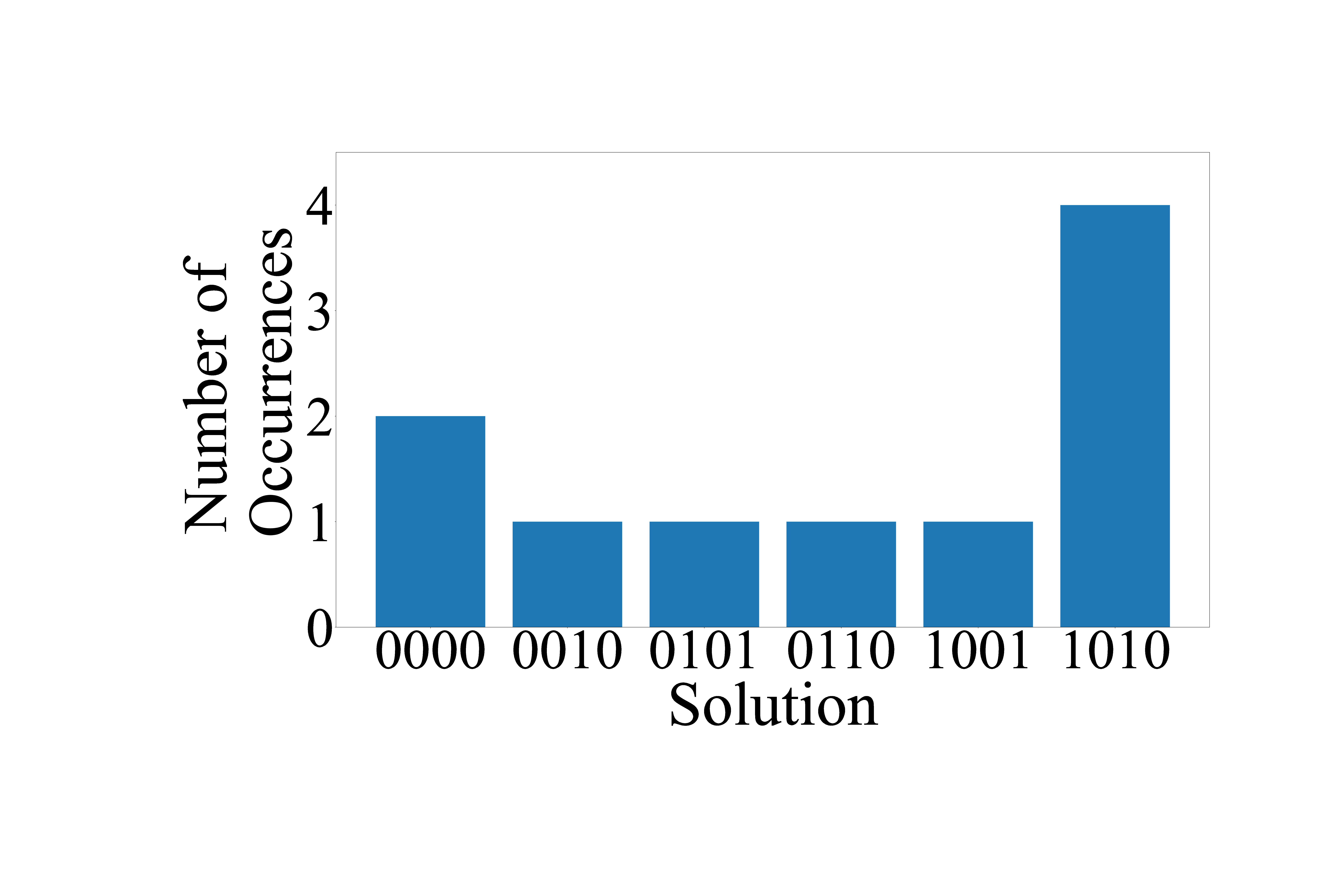}
  \caption{}
  \label{subfig:optimalsols_mbbstructure1}
\end{subfigure}%
\begin{subfigure}{0.58\textwidth}
  \centering 
  \includegraphics[width=\linewidth, trim=12in 7in -1.5in 0in, clip 
  ]{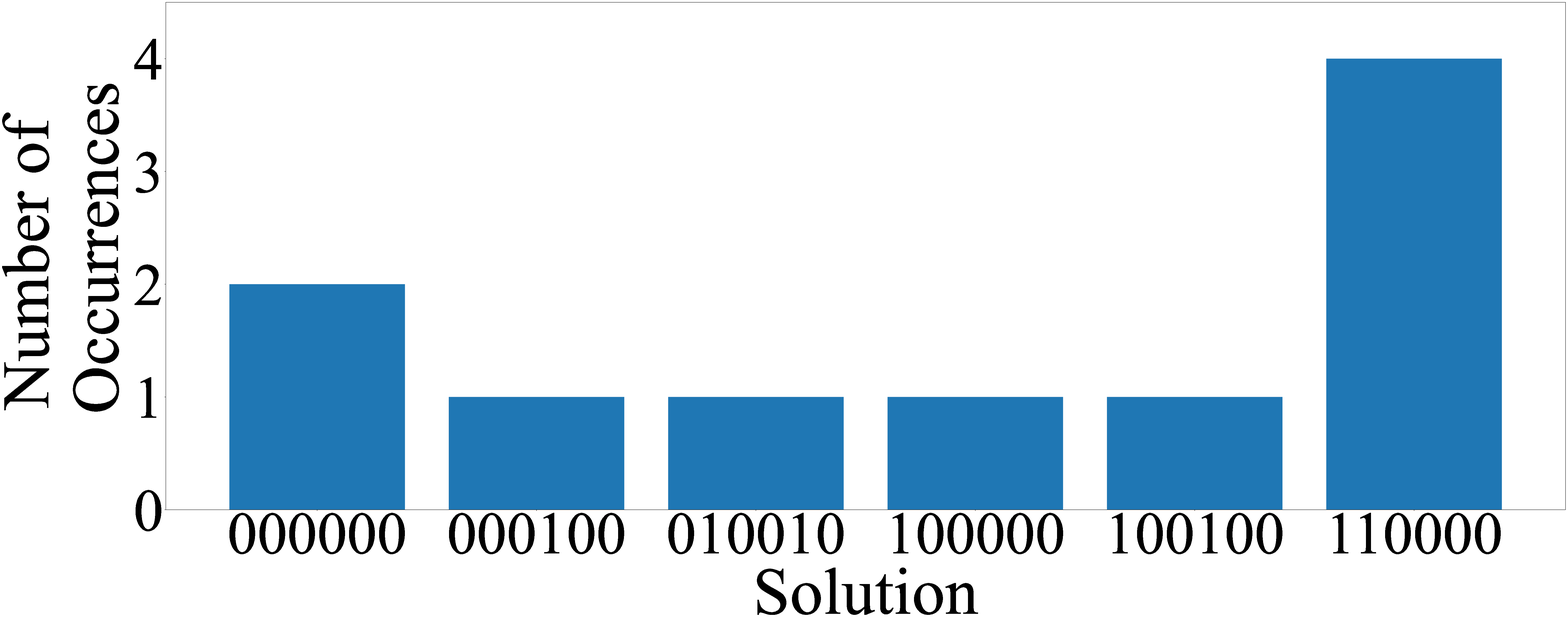}
  \caption{}
  \label{subfig:optimalsols_mbbstructure2}
\end{subfigure}
\caption{Distributions of the optimal configurations for MBB beams  (a) \#1 and (b) \#2 after 50 iterations}
\label{fig:mbbstructures_optimalsols}
\end{figure*}

Ten runs of the VQA are performed, where each run consists of 50 optimization iterations. The initial values of $\lambda$ and $\mu$ are $0.5$ and $5$, respectively. The values of the Lagrange multipliers in the beam examples are larger than those in the truss examples since it is more difficult to solve the PDE constraints. The value of $\lambda$ increases by $0.1$ after each iteration for both beams. After each iteration, the value of $\mu$ increases by $0.5$ and $0.1$ for beams \#1 and \#2, respectively. The change in $\mu$ is smaller for beam \#2 since the number of possible solutions is larger than that of beam \#1. 

Figure \ref{fig:mbbstructures_optimalsols} shows the distributions of optimal configurations. In Figure \ref{subfig:optimalsols_mbbstructure1}, four out of the runs result in $q_1q_2q_3q_4 = 1010$, which is the true optimal configuration for beam \#1 with volume constraints. In Figure \ref{subfig:optimalsols_mbbstructure2}, four runs result in $q_1q_2q_3q_4q_5q_6 = 110000$, which is the true optimal feasible configuration for beam \#2 with volume constraints. Overall, the success probability is 40\%. It is also observed that none of the optimal configurations contain more than two optional elements. Therefore, the VQA with the selected values of $\nu^*$ can  enforce the volume constraints when the number of optional elements is between four and six.

\section{Discussions and Conclusions} 
\label{sec:conclusions}

In this paper, a new variational quantum algorithm for constrained topology optimization is proposed to find the optimal material configuration and solve the PDE constraints simultaneously. The single-loop optimization is achieved by constructing the controlled problem Hamiltonian operators so that the two registers are entangled. The number of gates scales with a factor of $\mathcal{O}\left(v^2\log_2 \left({2}^{1/2} \Vert Q^{(\cdot)} \Vert_{\mathrm{max}} \right) \mathrm{poly}(n)\right)$. This factor depends on the sparsity of stiffness matrices, the degrees of element connectivity, and the number of qubits. Typically, the stiffness matrices are sparse and the degrees of element connectivity are small for two- and three-dimensional structures. Therefore, the number of gates scales polynomially with respect to the number of qubits.

The connectivity and volume constraints in the optimization problems are also incorporated in the proposed VQA. The probability upper bound of volume constraints is derived to provide the guidance for setting the upper bound of the rotation gate parameters. The advantage of this encoding scheme for volume constraints is its simplicity where the constraints on VQA parameters can be easily enforced on the classical computer side. Encoding the connectivity constraints, however, is more complex because circuits need to be customized with the entanglement between qubits. 

The proposed VQA is demonstrated with compliance minimization problems involving three truss structures and two MBB beams. The results show that the optimization is sensitive to the hyperparameters including the Lagrange multipliers and the Bayesian optimization settings. How to set up the hyperparameters to improve the fidelity and efficiency requires further studies. Across all five examples, the VQA is implemented with at most 12 qubits and 22 Hermitian operators in the problem Hamiltonian. Solving TO problems with larger structures requires more qubits and deeper circuits. Future work will also focus on reducing the number of qubits and the circuit depth of the VQA.

\bibliographystyle{unsrt}
\bibliography{references}

\end{document}